\documentclass[a4paper,twocolumn,notitlepage,nofootinbib,superscriptaddress]{revtex4-2}

\usepackage[english]{babel}
\usepackage{amsmath}
\usepackage{amssymb}
\usepackage{amsthm}
\usepackage{graphicx}
\usepackage{xcolor}
\usepackage{times}
\usepackage[pass]{geometry} 
\usepackage{mathrsfs}
\usepackage{enumitem}
\usepackage{csquotes}
\usepackage{braket}
\usepackage{mathtools}
\usepackage{booktabs}
\usepackage{thmtools}
\usepackage{enumitem}
\usepackage{algorithm}

\usepackage{algpseudocode}
\usepackage{verbatim}
\usepackage{soul}
\usepackage{apptools}
\usepackage[most]{tcolorbox}
\usepackage{xcolor}
\usepackage[bookmarks=true,colorlinks=true,linkcolor=teal,citecolor=teal,urlcolor=teal,bookmarksopen=true,bookmarksopenlevel=1]{hyperref}
\usepackage[capitalize]{cleveref}

\newcommand{\Id}{\ensuremath{\mathbb{I}}}

\newcommand{\norm}[1]{\left\lVert#1\right\rVert}

\providecommand{\myvec}[1]{\ensuremath{\boldsymbol{#1}}}

\providecommand{\bb}{\ensuremath{\myvec{b}}}

\providecommand{\rr}{\ensuremath{\myvec{r}}}

\providecommand{\uu}{\ensuremath{\myvec{u}}}
\providecommand{\vv}{\ensuremath{\myvec{v}}}

\providecommand{\xx}{\ensuremath{\myvec{x}}}
\providecommand{\yy}{\ensuremath{\myvec{y}}}

\providecommand{\calH}{\ensuremath{\mathcal{H}}}
\providecommand{\calI}{\ensuremath{\mathcal{I}}}

\providecommand{\calO}{\ensuremath{\mathcal{O}}}
\providecommand{\calP}{\ensuremath{\mathcal{P}}}

\providecommand{\calS}{\ensuremath{\mathcal{S}}}

\providecommand{\bbP}{\ensuremath{\mathbb{P}}}

\providecommand{\bbZ}{\ensuremath{\mathbb{Z}}}

\newcommand{\tvert}[1]{{\left\vert\kern-0.25ex\left\vert\kern-0.25ex\left\vert #1 
    \right\vert\kern-0.25ex\right\vert\kern-0.25ex\right\vert}}

\newcommand{\poly}{\ensuremath{\operatorname{poly}}}
\newcommand{\proj}[1]{\ket{#1}\bra{#1}}
\def\ketbra#1#2{{\vert#1\rangle\!\langle#2\vert}} 
\newcommand{\Abs}[1]{\left\lvert #1 \right\rvert}
\newcommand{\True}{\texttt{TRUE}}
\newcommand{\False}{\texttt{FALSE}}

\let\oldproof\proof
\renewcommand{\proof}{\color{black}\oldproof}
\DeclareMathOperator*{\Argmax}{arg\,max}

\DeclarePairedDelimiterX{\infdivx}[2]{[}{]}{%
  #1\;\delimsize\|\;#2%
}
\newcommand{\infdiv}{D\infdivx}

\tcbset {
  base/.style={
    arc=0mm, 
    bottomtitle=0.5mm,
    boxrule=0mm,
    colbacktitle=black!10!white, 
    coltitle=black, 
    fonttitle=\bfseries, 
    left=2.5mm,
    leftrule=1mm,
    right=3.5mm,
    title={#1},
    toptitle=0.75mm, 
  }
}

\definecolor{brandblue}{rgb}{0.34, 0.7, 1}
\newtcolorbox{mainbox}[1]{
  colframe=brandblue, 
  base={#1}
}

\newtcolorbox{subbox}[1]{
  colframe=black!30!white,
  base={#1}
}

\AtBeginDocument{
\heavyrulewidth=.08em
\lightrulewidth=.05em
\cmidrulewidth=.03em
\belowrulesep=.65ex
\belowbottomsep=0pt
\aboverulesep=.4ex
\abovetopsep=0pt
\cmidrulesep=\doublerulesep
\cmidrulekern=.5em
\defaultaddspace=.5em
}

\declaretheoremstyle[
  headfont=\color{red}\normalfont\bfseries,
  bodyfont=\color{red}\normalfont\itshape,
]{colored}

\newtheorem{theorem}{Theorem}
\newtheorem{lemma}[theorem]{Lemma}
\newtheorem{proposition}[theorem]{Proposition}
\newtheorem{definition}[theorem]{Definition}
\newtheorem{corollary}[theorem]{Corollary}

\newtheorem*{oproblem}{Open problem}
\newtheorem*{remark}{Remark}

\AtAppendix{\counterwithin{stheorem}{section}}
\AtAppendix{\counterwithin{slemma}{section}}
\AtAppendix{\counterwithin{sproposition}{section}}
\AtAppendix{\counterwithin{sdefinition}{section}}
\AtAppendix{\counterwithin{scorollary}{section}}
\AtAppendix{\counterwithin{sconjecture}{section}}
\AtAppendix{\counterwithin{sproblem}{section}}

\definecolor{alexcolor}{HTML}{e74c3c}

\NewDocumentCommand{\lr}{+!o +!d() +!d|| +!d<>}

\newcommand{\dccqs}{Dahlem Center for Complex Quantum Systems, Freie Universit{\"a}t Berlin, 14195 Berlin, Germany}
\newcommand{\hzb}{Helmholtz-Zentrum Berlin f{\"u}r Materialien und Energie, 14109 Berlin, Germany}
\newcommand{\hhi}{Fraunhofer Heinrich Hertz Institute, 10587 Berlin, Germany}

\newcommand{\papertitle}{A measurement-driven quantum algorithm for SAT:\\Performance guarantees via spectral gaps and measurement parallelization}

\begin{document}	

\title{\papertitle}
\date{\today}

\author{Franz\ J.\ Schreiber}
\thanks{These authors contributed equally to this work.\\ \href{mailto:f.schreiber@fu-berlin.de}{f.schreiber@fu-berlin.de}, \href{mailto:m.kramer@fu-berlin.de}{m.kramer@fu-berlin.de}.}
\affiliation{\dccqs}

\author{Maximilian\ J.\ Kramer}
\thanks{These authors contributed equally to this work.\\ \href{mailto:f.schreiber@fu-berlin.de}{f.schreiber@fu-berlin.de}, \href{mailto:m.kramer@fu-berlin.de}{m.kramer@fu-berlin.de}.}
\affiliation{\dccqs}

\author{Alexander\ Nietner}
\affiliation{\dccqs}

\author{Jens\ Eisert}
\affiliation{\dccqs}
\affiliation{\hzb}
\affiliation{\hhi}

\begin{abstract}
The Boolean satisfiability problem (SAT) is of central importance in both theory and practice. Yet, most provable guarantees for quantum algorithms rely exclusively on Grover-type methods that cap the possible advantage at only quadratic speed-ups, making the search for approaches that surpass this quadratic barrier a key challenge. In this light, this work presents a rigorous worst-case runtime analysis of a recently introduced measurement-driven quantum SAT solver. Importantly, this quantum algorithm does not exclusively rely on Grover-type methods and shows promising numerical performance. Our analysis establishes that the algorithm's runtime depends on an exponential trade-off between two key properties: the spectral gap of the associated Hamiltonian and the success probability of the driving measurements. We show that this trade-off can be systematically controlled by a tunable rotation angle. Beyond establishing a worst-case runtime expression, this work contributes significant algorithmic improvements. First, we develop a new readout routine that efficiently finds a solution even for instances with multiple satisfying assignments. Second, a measurement parallelization scheme, based on perfect hash families, is introduced. Third, we establish an amplitude-amplified version of the measurement-driven algorithm. Finally, we demonstrate the practical utility of our framework: By suitably scheduling the algorithm's parameters, we show that its runtime collapses from exponential to polynomial on a special class of SAT instances, consistent with their known classical tractability. A problem we leave open is to establish a non-trivial lower bound on the spectral gap as a function of the rotation angle. Resolving this directly translates into an improved worst-case runtime, potentially realizing a super-quadratic quantum advantage.
\end{abstract}

\maketitle

\section{Introduction}

The \emph{Boolean satisfiability problem} (SAT) is one of \emph{the} central problems in computer science, holding foundational importance for both complexity theory and practical optimization. Since the landmark Cook-Levin theorem established 
its $\mathsf{NP}$-completeness~\cite{Cook1971, Levin1973}, 
SAT---and $3$-SAT in particular---has served as a canonical benchmark for computational hardness. The immense practical value of solving large SAT instances, combined with its established intractability, makes SAT a prime target for novel computational paradigms, thereby motivating an urgent search for more efficient quantum algorithms.

The \emph{exponential time hypothesis} (ETH) posits an exponential worst-case runtime for SAT \cite{impagliazzo2001eth1, impagliazzo2001eth2}. As a consequence, an exponential separation between the run times of quantum and classical solvers is believed to be impossible. Adding to that, the \emph{quantum strong exponential time hypothesis} (QSETH) \cite{buhrman2021_quantum_SETH, aaronson2020quantumcomplexityclosestpair} suggests that quantum speed-ups for SAT are at most quadratic. While the accepted complexity-theoretic assumptions leave open the compelling possibility of a \emph{super-quadratic quantum speed-up} for $3$-SAT, most known approaches rely on straightforward Grover-type adaptations of classical algorithms \cite{grover1996, Brassard2002, ambainis2004, dantsin2005_groverization} or utilize quantum walk techniques \cite{montanaro2016backtracking, Ambainis2017}. Consequently, the achievable speed-up is capped at a quadratic improvement over the best classical runtime. Pursuing such a super-quadratic advantage is especially relevant in light of the postulated ``quartic barrier'' to practical quantum advantage, which suggests a merely quadratic speed-up may be insufficient to overcome the overheads of error correction \cite{beyondQuadratic} in the foreseeable future. 

In this work, we take a substantial step in that direction by providing a rigorous analysis of a novel measurement-driven quantum algorithm introduced in Ref.~\cite{benjamin2017}. 
Besides being hardware-friendly on \emph{near-term quantum devices}, 
it uses inherently quantum features by introducing a key parameter, the rotation angle $\theta$, which systematically tunes the degree of non-commutativity of the measurements that drive the computation. 
The regime of \emph{near-term intermediate scale quantum} (NISQ) computers augmented by mid-circuit measurements has been dubbed \emph{NISQ+} \cite{NISQPlus}: As an important milestone towards \emph{fault-tolerant application-scale quantum} (FASQ) computers \cite{MindTheGaps}, the algorithm at hand seems to fit in this highly relevant regime that can be accommodated in several physical platforms with a special emphasis on being particularly hardware-friendly on platforms with all-to-all connectivity that implement the quantum measurements used in our setting very natively (such as photonic \cite{bartolucci2021fusionbasedquantumcomputation,bombin2021}, cold atom \cite{evered2023high,bluvstein2022quantum,bluvstein2024logical}, and ion-trap \cite{PhysRevX.13.041052,kielpinski2002architecture} devices). In this work, we make \emph{conceptual} as well as \emph{algorithmic} contributions. 

On the conceptual side, we establish the dependence of the worst-case runtime on two key, tunable properties: the spectral gap of the associated Hamiltonian and the success probability of the driving measurements. Using the \emph{method of alternating projections} \cite{escalante2011alternating} and the \emph{detectability lemma} \cite{aharonov2011DL,Anshu_2016}, our analysis in \cref{sec:Measurement-driven quantum SAT solver} reveals an exponential trade-off between these two properties, which is controlled by the algorithm's rotation angle, $\theta$. In contrast to the undecidability of spectral gaps in general \cite{cubitt2022}, the Hamiltonians considered in this work are frustration-free and have a particular structure that can be exploited to make statements about the gap scaling (see \cref{subsubsec:rate of convergence,sec:appendix:Detectability Lemma and Method of Alternating Projections,sec:worst-case_bounds_for_gaps}).

We then demonstrate the practical utility of this framework: In \cref{sec:Analysis of some restricted input classes}, we show that the algorithm's runtime on certain inputs can be exponentially improved compared to a naive parameter setting. On these instances, appropriately tuning $\theta$ according to our analysis gives rise to a polynomial time algorithm, matching their classical tractability.

On the algorithmic side, we significantly broaden the scope, practicality, and performance of the algorithm. In \cref{subsec:time complexity of state preparation}, we point out that the algorithm is compatible with amplitude-amplification techniques, boosting the asymptotic runtime. In \cref{subsec:inferring_solution}, we introduce a new, rigorous readout routine capable of efficiently finding a solution even for general instances with multiple satisfying assignments. Furthermore, in \cref{subsec:parallelized_measurements} we develop a measurement parallelization scheme using a perfect hash family construction. The resulting \emph{layers} of commuting measurements can each be implemented as a single, hardware-friendly measurement, severely reducing the overall runtime of our algorithm.

On a general note, the methods and results of this work offer insights of potential independent interest for the broader analysis of dissipation-driven algorithms.

\section{Background}
\subsection{The Boolean satisfiability problem}
The \emph{Boolean satisfiability problem} (SAT) asks us to decide whether a 
given Boolean formula $\phi$ admits a truth value assignment $\xx \in \{0,1\}^n$ to the $n$ variables such that the formula, as a whole, evaluates to true under this assignment. This formula $\phi$ consists of $n$ Boolean variables $b_1, \dots, b_n$ that are connected by the Boolean operators (conjunction $\wedge$, disjunction $\vee$, and negation $\neg$). For every such formula, there is an equivalent formula that is in \emph{conjunctive normal form} (CNF), i.e., is a conjunction of $m$ clauses,
\begin{align}
    \phi = (c_1 \wedge c_2 \wedge \dots \wedge c_m)
\end{align}
with $c_j = (l_{j_1} \vee l_{j_2} \vee \dots \vee l_{j_k})$. Here, each clause $c_j$ is a disjunction of at most $k$ literals. A literal is a stand-in for either a variable $b_i$ or its negation $\neg b_i$. $k$-SAT is known to be in $\mathsf{P}$ for $k=2$ \cite{schaefer1978} and is $\mathsf{NP}$-complete for $k\geq 3$ due to the famous Cook–Levin theorem~\cite{Cook1971, Levin1973} which makes it relevant for theoretical computer science. Of particular interest is $3$-SAT as it is the canonical $\mathsf{NP}$-complete problem to which many other problems are reduced. Moreover, $3$-SAT serves as a central benchmark in both theoretical and practical studies, with widespread applications in areas such as hardware and software verification, planning and scheduling, and cryptographic problem solving. Accordingly, our review and analysis below will put a special emphasis on $3$-SAT.

\subsection{Classical SAT solvers}
Classical $3$-SAT solvers have been studied extensively, with worst-case runtimes available for several prominent algorithms. The naive brute-force search runs in $\mathcal{O}(2^n)$, but more advanced algorithms, often exploiting $3$-local structures, achieve much better runtimes. One such notable example is Schöning's stochastic local search algorithm \cite{Schoening1999}, known for its simplicity, which achieves a runtime of $\mathcal{O}(1.3334^n)$. The currently best rigorously proven runtime is achieved by a variant of the PPSZ algorithm \cite{PPSZ_2005,Hertli_PPSZ}, which achieves a runtime of $\mathcal{O}(1.306995^n)$ \cite{Hansen_biased_PPSZ}. A comprehensive overview of runtime guarantees for $3$-SAT solvers can be found in Refs.~\cite{schöning2013satisfiability,Biere2021}.

Modern SAT solvers, which utilize heuristics, have become highly effective by combining techniques such as \emph{backtracking}~\cite{Davis1960,Davis1962}, \emph{conflict-driven clause learning} (CDCL)~\cite{Grasp_CDCL, Chaff_CDCL}, and \emph{inprocessing}. These methods enable them to efficiently solve large, structured instances from various applications~\cite{schöning2013satisfiability, Biere2021}. The gap between this stellar performance of modern solvers on instances of interest and their worst-case exponential runtime remains poorly understood~\cite{Ganesh2020_unreasonable_effectiveness_SAT}.

\subsection{Quantum SAT solvers}
A natural approach to solving SAT on a quantum computer is to apply Grover's algorithm \cite{grover1996} to achieve a quadratic speed-up over brute-force search, i.e., $\mathcal{O}(1.414^n)$. Further improvements are possible by ``groverizing'' (via amplitude amplification \cite{Brassard2002}) more sophisticated classical algorithms that utilize the instance's structure. Here, we focus on results for $3$-SAT. For example, Ref.~\cite{ambainis2004} demonstrates that this framework can be utilized to achieve a quadratic improvement over Schöning's algorithm, which yields a quantum runtime of $\mathcal{O}(1.155^n)$. However, a quadratic improvement over the record-breaking PPSZ algorithm~\cite {PPSZ_2005,Hertli_PPSZ} can also be achieved, as discussed in Ref.~\cite{dantsin2005_groverization}, thus yielding an asymptotic runtime of $\mathcal{O}(1.144^n)$. Moreover, hybrid quantum-classical variants of Schöning's algorithm have been investigated more recently in Refs.~\cite{Dunjko2018,eshaghian2024_hybrid_schoening}.

Quantum backtracking algorithms~\cite{montanaro2016backtracking, Ambainis2017} provide a quantum approach to solving SAT on a quantum computer by accelerating classical backtracking procedures on which many state-of-the-art classical SAT solvers rely. Unlike Grover-based approaches, these methods utilize a quantum walk framework to explore the backtracking tree, offering an almost quadratic speed-up in terms of the number of nodes visited. Suppose a classical algorithm explores $T$ nodes in the search tree. In that case, there is a quantum backtracking algorithm that decides whether the instance is satisfiable or unsatisfiable in $\mathcal{O}(\poly(n)\sqrt{T})$ steps. 
It should be noted that $T$ scales usually exponentially in the number of variables, i.e., $n$. 
Refs.~\cite{Campbell2019, brehm2024} highlight the importance of evaluating these approaches beyond worst-case asymptotics and report actual scalings and performance benchmarks on structured, practical instances.

Adiabatic quantum computing and quantum annealing~\cite{kadowaki1998, farhi2000, kadowaki2002} are popular heuristics for combinatorial optimization, and as such also for SAT, but suffer from exponentially small gaps~\cite{vanDam2001, farhi2002, Reichardt2004, farhi2005, Znidaric2006, farhi2010, Altshuler2010, hen2014, werner2023, braida2025}. Inspired by annealing, the \emph{quantum approximate optimization algorithm} (QAOA)~\cite{Montanaro_SAT_with_QAQO, montanaro2024QAOA} and its Grover-enhanced variant~\cite{zhang2024, Mandl_2024} show promise for $k$-SAT, though rigorous guarantees remain often lacking.

Another line of work is centered around \emph{dissipation-driven quantum computing}, first introduced in Refs.~\cite{Childs2002, Verstraete2009,Kraus2008}. This paradigm, which leverages engineered dissipation as a computational resource, has recently attracted considerable attention. Current research includes advances in dissipative ground state preparation (see, e.g., Refs.~\cite{Gilyen2017,cubitt2023dissipative,chen2023,ding2024singleancillagroundstatepreparation,Lambert2024,li2024,Mi2024,motlagh2024,Eder2025,zhan2025,Lloyd2025,purcell2025}) and Gibbs sampling (see, e.g., Refs.~\cite{chen2023efficientexactnoncommutativequantum,chen2023fastthermalizationeigenstatethermalization,chen2023quantumthermalstatepreparation, zhang2023dissipativequantumgibbssampling,Ding2025, gilyen2024quantumgeneralizationsglaubermetropolis,jiang2024quantummetropolissamplingweak,Mozgunov2020,Rall2023,shtanko2023,Temme2011}).
The ground state preparation techniques are, in principle, applicable to combinatorial optimization problems such as SAT. For these problems, the corresponding Hamiltonian would (in the simplest case) consist of mutually commuting terms that are diagonal in the computational basis. While some dissipative approaches were developed precisely to handle arbitrary non-commuting Hamiltonians \cite{cubitt2023dissipative} and offer convergence guarantees under global resampling strategies, their worst-case runtime still scales exponentially as $\mathcal{O}(2^n)$. Applied to SAT, this performance is as bad as classical brute-force search. The algorithm investigated in this work is also dissipation-driven. However, it relies on a different mechanism: it combines discrete, projective measurements with a clever encoding of SAT into a Hamiltonian that introduces non-commutativity in a very controlled fashion while staying frustration-free.

\subsection{On the quantum-classical separation for SAT}
Before turning to the algorithm's details, we first discuss what types of quantum speed-ups can reasonably be expected for SAT. The widely accepted \emph{exponential time hypothesis} (ETH) suggests that even quantum computers are unlikely to solve SAT in sub-exponential time in the worst case. In fact, the \emph{quantum strong exponential time hypothesis} (QSETH) implies that Grover's search, which runs in $\mathcal{O}(2^{n/2})$, represents the optimal quantum speed-up \cite{buhrman2021_quantum_SETH}, compared to the classical $\mathcal{O}(2^n)$ scaling implied by the \emph{strong exponential time hypothesis} (SETH) \cite{impagliazzo2001eth1, impagliazzo2001eth2}.

However, if we restrict our attention to $3$-SAT on $n$ variables, only the ETH \cite{impagliazzo2001eth1, impagliazzo2001eth2} is relevant. It states that the worst-case classical runtime will be $\mathcal{O}(\lambda_{\mathrm{c}}^n)$ with $\lambda_{\mathrm{c}}>1$. In fact, as pointed out above, the best-known classical algorithm achieves $\lambda_c^* = 1.306995$ \cite{Hansen_biased_PPSZ}. For quantum computers, we also expect $\mathcal{O}(\lambda_{\mathrm{q}}^n)$ with $\lambda_{\mathrm{q}}>1$. However, complexity-theoretic assumptions do not rule out the possibility of $\lambda_{\mathrm{q}} < \sqrt{\lambda_{\mathrm{c}}^*}$, thus yielding a super-quadratic quantum speed-up for $3$-SAT.

While the above-mentioned complexity-theoretic arguments rule out super-polynomial quantum advantages in the worst-case, they leave open the possibility of such advantages for specific subsets of instances. Only recently has this perspective begun to receive systematic study (see, e.g., Refs.~\cite{Pirnay2024,szegedy2022,Yamakawa2024,jordan2024DQI,buhrman2025formalframeworkquantumadvantage}).

\subsection{Related works}
The present work is an extension of the measurement-driven quantum algorithm introduced in Ref.~\cite{benjamin2017} (with details provided in \cref{subsec:encoding_hamiltonian,subsec:algorithm_overview}). In the past, this work has already been extended in two directions: 
On the one hand, Ref.~\cite{Zhao2019measurement_driven} extends the framework of Ref.~\cite{benjamin2017} to general frustration-free Hamiltonians. The authors describe a measurement-driven analog of adiabatic quantum computation for frustration-free Hamiltonians. Here, slowly varying measurements are executed to mimic the adiabatic evolution. It is shown that for Hamiltonians that remain frustration-free along the evolution path, the necessary measurements can be implemented using measurements of random terms of the Hamiltonian. The presented results are based on a connection between the \emph{adiabatic theorem} and the \emph{quantum Zeno effect}.
On the other hand, Ref.~\cite{zhang2024ksatproblems} extends the original projection-based measurement approach of Ref.~\cite{benjamin2017} to quantum measurements of arbitrary strength, a technique often referred to as \emph{Zeno dragging}. Subsequently, Ref.~\cite{zhang2025} offers a more detailed analysis of this generalized measurement framework, including a partial analytical explanation of numerical observations reported in earlier work.
Moreover, it is worth noting that the present work can be cast within the dissipative framework introduced in Ref.~\cite{cubitt2023dissipative}.

\section{Measurement-driven quantum SAT solver} \label{sec:Measurement-driven quantum SAT solver}
A novel quantum algorithm for solving SAT has been proposed in Ref.~\cite{benjamin2017} (based on an earlier manuscript to be found in Ref.~\cite{benjamin2015}). Promising numerics indicate that this algorithm outperforms Schöning's algorithm~\cite{Schoening1999} in solving certain $3$-SAT instances. The proposed quantum algorithm consists of repeated projective measurements, 
where each measurement corresponds to the truth value of a generalized clause. The details are explained below.
\subsection{Encoding SAT into a rotated Hamiltonian}
\label{subsec:encoding_hamiltonian}
We encode a SAT instance $\phi(n,m)$ in CNF with $n$ variables and $m$ clauses into a Hamiltonian. As such, determining a zero-energy ground state of the Hamiltonian is equivalent to finding a satisfying solution to the SAT instance. 
First, we choose a mapping from truth value assignments to quantum states. Solution candidates for $\phi(n,m)$ are represented as binary strings of length $n$, encoding \True{} and \False{} assignments of the $n$ variables in Boolean logic. A typical way to encode such a length $n$ binary string into an $n$-qubit quantum state is to identify the Boolean values $\{0,1\}$ with the computational basis state vectors $\{\ket{0},\ket{1}\}$. Here, we deviate from this paradigm and parametrize our encoding with a tunable parameter $\theta \in (0, \frac{\pi}{2}]$, which we refer to as the \emph{rotation angle}. The parametrization of the rotated normalized basis states is chosen such that the standard encoding 
is recovered in the limit $\theta=\frac{\pi}{2}$
as
\begin{align}
    \mathrm{\True{}} = \ket{1} & \mapsto \ket{\theta} = R_Y(+\theta) \ket{+},\\
    \mathrm{\False{}} = \ket{0} & \mapsto \ket{\bar{\theta}} = R_Y(-\theta) \ket{+},
\end{align}
where $R_Y(\theta)$ is the usual rotation operator around the Y-axis, given as
\begin{align}
    R_Y(\theta) = 
    \begin{pmatrix}
        \cos(\theta/2) & -\sin(\theta/2)\\
        \sin(\theta/2) & \cos(\theta/2)
    \end{pmatrix}.
\end{align}
By way of example, the all-zero binary string is then encoded as 
$0^n 
\mapsto \ket{\bar{\theta}}^{\otimes n}$. We illustrate 
this encoding in the top of \cref{fig:rotated_states}.
The Hamiltonian corresponding to $\phi(n,m)$ takes the general form 
\begin{align}\label{eqn:SAT encoding Hamiltonian}
    H(\theta) = \sum_{i=1}^m P_i(\theta),
\end{align}
where each $P_i(\theta)$ is an orthogonal projector corresponding to the $i$'th of the $m$ clauses. We denote the quantum state encoding the Boolean truth value assignment $\xx \in \{0,1\}^n$ via a tensor product of $\ket{\theta}$- and $\ket{\bar{\theta}}$-state vectors as $\ket{\Theta_{\xx}}$. More formally, we define
\begin{align} \label{eqn:theta string}
    \ket{\Theta_{\xx}} \coloneqq \bigotimes_{i=1}^n R_Y \left( L_i \theta \right) \ket{+},
\end{align}
where $L_i=1$ if $x_i=1$ and $L_i=-1$ if $x_i=0$. 
The projector $P_i(\theta)$ is constructed such that $\braket{\Theta_{\xx}|P_i(\theta)|\Theta_{\xx}}= 0$ if $\xx$ satisfies the $i$'th clause and $\braket{\Theta_{\xx}|P_i(\theta)|\Theta_{\xx}}\neq0$ else. To this end, we introduce the pair of normalized
state vectors 
\begin{align}
    \ket{\theta^\perp} &= R_Y(\pi + \theta) \ket{+}, \\
    \ket{\bar{\theta}^\perp} &= R_Y(\pi - \theta) \ket{+},
\end{align}
with the properties
\begin{align}
\braket{\theta | \theta^\perp} &= 0, & \braket{\bar{\theta} | \bar{\theta}^\perp} &= 0, \\
\braket{\theta | \bar{\theta}} &= \cos(\theta), & \braket{\theta^\perp | \bar{\theta}^\perp} &= \cos(\theta), \\
\braket{\theta | \bar{\theta}^\perp} &= \sin(\theta), & \braket{\bar{\theta} | \theta^\perp} &= \sin(\theta).
\end{align}

Generalizing the above, we denote by $\ket{\Theta_{\xx}^\perp}$ the state vector perpendicular to $\ket{\Theta_{\xx}}$ which results from changing $R_Y(L_i\theta)$ to $R_Y(\pi + L_i\theta)$ in \cref{eqn:theta string}.
The projector $P_i(\theta)$ is then constructed in the natural way: A $k$-local clause $c_i$ has the effect of disallowing any potential solution string containing a certain $k$-local binary string $\xx_i=(x_{i_1},x_{i_2},\dots,x_{i_k})$. The corresponding clause projector $P_i(\theta)$ is then
\begin{align}
    P_i(\theta) = \proj{\Theta_{\xx_i}^\perp }_{i_1,i_2,\dots,i_k} \otimes \Id_{[n]\setminus \{i_1,i_2,\dots,i_k\}}.
\end{align}
For example, the $3$-SAT clause $c_i = (b_{1} \vee b_{4} \vee \neg b_6)$ corresponds to
\begin{align}
    P_i(\theta) = \proj{\theta_{1}^\perp \theta_{4}^\perp \bar{\theta}_{6}^\perp} \otimes \Id_{[n]\setminus \{1,4,6\} }.
\end{align}
Let us denote by $\calS$ the set of solution-encoding binary strings for a given SAT-instance.
A quantum state vector $\ket{\Theta_{\xx}}$ encoding a satisfying solution lives in the kernel of all projectors $P_i(\theta)$, such that $H(\theta)\ket{\Theta_{\xx}}=0$ for a binary string $\xx \in \calS$. In fact, the complete ground space of $H(\theta)$ is spanned by solution states, meaning any state in the ground space of $H(\theta)$ can be written as a superposition of state vectors $\ket{\Theta_{\xx}}$ encoding the satisfying solutions of the underlying SAT formula $\phi(n,m)$. This follows from Ref.~\cite[Appendix F]{benjamin2017} where it is proven that the rotation preserves the original ground state dimension, i.e., $d_{\mathrm{sol}}$ classical solutions imply $d_{\mathrm{sol}}$ quantum solutions for $0<\theta<\frac{\pi}{2}$.

\begin{figure}
    \centering
    \includegraphics[width=0.99\linewidth]{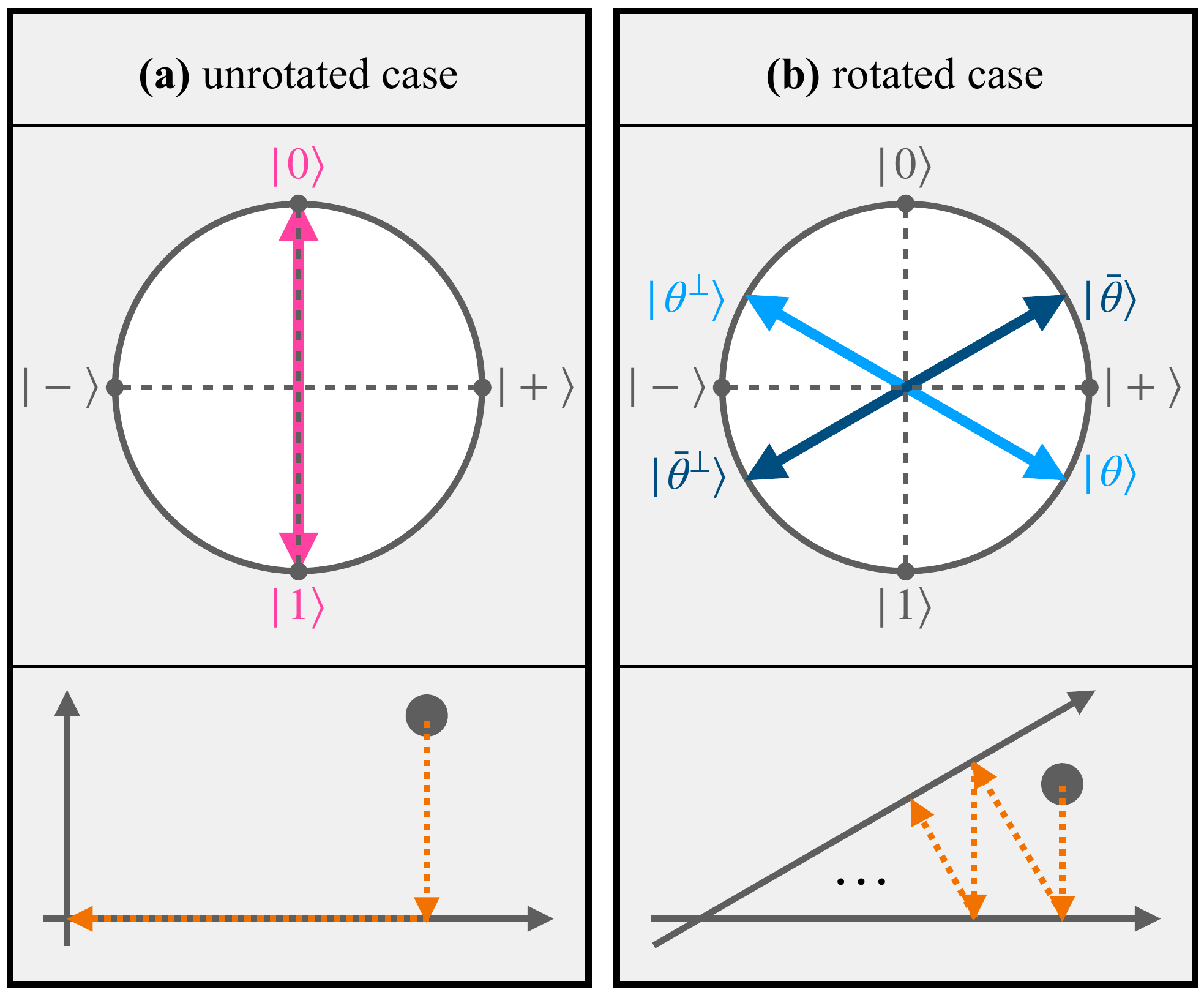}
    \caption{Illustration of (a) the unrotated and (b) the rotated settings. The top panels show the locations of the corresponding states in the XZ-projection of the Bloch sphere. The bottom panels illustrate, for two projectors, how convergence proceeds in each setting. Using an orthogonal encoding ($\theta=\frac{\pi}{2}$), we converge to the ground space at $(0,0)$ with a single pass of each clause check. In the non-orthogonal setting with $\theta \neq \frac{\pi}{2}$, we slowly converge towards the ground space.} 
    \label{fig:rotated_states}
\end{figure}

\subsection{Overview of the algorithm} \label{subsec:algorithm_overview}

In this section, we will elaborate on the actual measurement-assisted algorithm. On a high level, we find a satisfying solution for the SAT formula $\phi(n,m)$ by iterating a two-stage process: a \emph{state preparation routine} followed by a \emph{solution readout routine}. First, run a \emph{state preparation routine} that produces a state vector $\ket{\psi_{\mathrm{out}}}$ that has a high fidelity with the ground space of $H(\theta)$. Since the ground space is spanned by all states that encode satisfying assignments, $\ket{\psi_{\mathrm{out}}}$ serves as a quantum representation of the solution set. In the second step, we perform a \emph{specialized readout measurement} on $\ket{\psi_{\mathrm{out}}}$ to extract a single classical solution string from the rotated ground space. This two-stage procedure is repeated sufficiently often to ensure that a valid satisfying assignment is found with high probability. The overall scheme is summarized in \cref{alg:quantum SAT solver}.

\begin{algorithm}[H]
\caption{Measurement-driven quantum SAT solver}
    \label{alg:quantum SAT solver}
    \begin{algorithmic}
    \Require SAT formula $\phi(n,m)$ with $n$ variables and $m$ clauses, angle $\theta$, failure probability $\delta$
    \Ensure satisfying assignment $\myvec{s}$
    \State \hrulefill
    \State $\myvec{s} \gets \myvec{0}_n$ \Comment{$\myvec{s} \in \{0,1\}^n$, stores solution}
    \State $\epsilon \gets \epsilon(\theta)$ \Comment{tolerance, see \cref{thrm:unique_solution_readout,thrm:multiple_solution_readout}}
    \State $R \gets R(\theta,n,\delta)$ \Comment{\#copies, see \cref{thrm:unique_solution_readout,thrm:multiple_solution_readout}}
    \For{$i$ in $1:R$}
    \State prepare $\ket{\psi_\mathrm{out}}$ with tolerance $\epsilon$ \Comment{see \cref{alg:state preparation}}
    \State $\myvec{s} \gets \textsc{Readout}\left(\ket{\psi_\mathrm{out}}\right)$ \Comment{see \cref{alg:unique_solution_readout,alg:multiple_solution_readout}}
    \EndFor
    \end{algorithmic}
\end{algorithm}

\subsubsection{State preparation routine}

\begin{figure}[t]
    \begin{center}
    \includegraphics[width=0.55\linewidth]{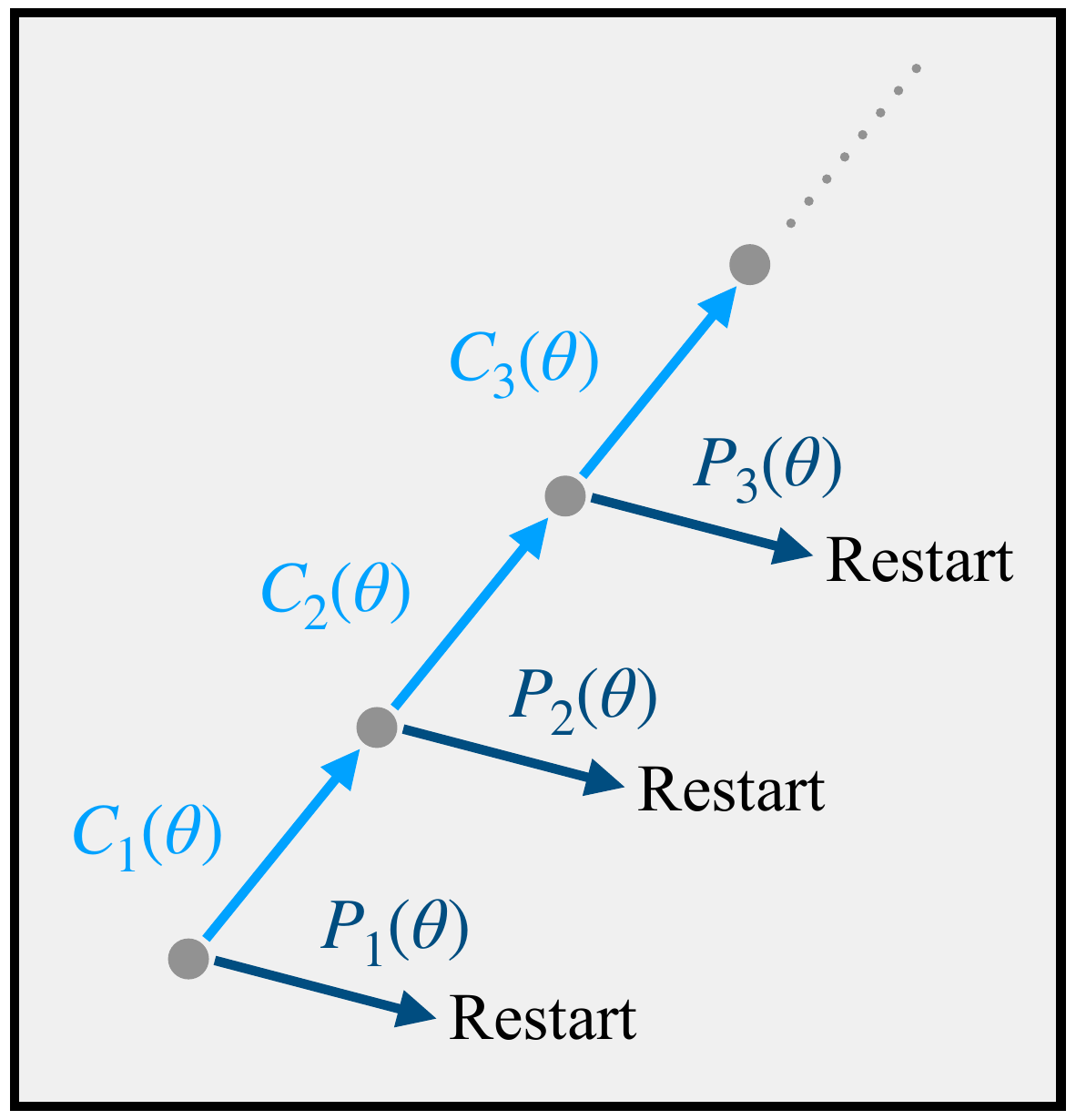}
    \caption{Algorithmic primitive of the measurement-driven approach. We sequentially perform measurements $\{C_i(\theta), P_i(\theta)\}$, with the measurement outcomes associated with $C_i(\theta)$ driving the state in the desired direction. Whenever we encounter an undesired outcome, we restart the procedure. The procedure stops 
    once we are sufficiently close to our target state.}
    \label{fig:algorithmic_primitive}
    \end{center}
\end{figure}

We begin the state preparation routine with the initial state vector $\ket{\psi_0} = \ket{+}^{\otimes n}$. It is worth noting at this 
stage that this state 
is an equal superposition of all basis state vectors 
$\ket{\Theta_{\xx}}$, i.e.,
\begin{align}
    \ket{\psi_0} = \ket{+}^{\otimes n} = \frac{1}{(2(1 + \cos(\theta)))^{n/2}} \sum_{\xx \in \{0,1\}^n} \ket{\Theta_{\xx}}.
\end{align}
By this rewriting, we see that an overlap between the ground space and the initial state is guaranteed. Further, let $C_i(\theta) \coloneqq \Id-P_i(\theta)$, such that $\{C_i(\theta), P_i(\theta)\}$ defines a projective measurement. These three-local measurements can be implemented in the circuit model with a triply-controlled-NOT gate and single qubit rotations \cite{benjamin2017}. We call such a measurement a \emph{clause check}. If we observe the measurement outcome associated with $C_i(\theta)$, we call the clause check \emph{passed}, else we say the clause check is \emph{failed}. We now drive $\ket{\psi_0}$ towards the ground state of $H(\theta)$ by sequentially performing clause checks on the state, continuing conditioned on passing the clause checks and restarting if a clause check fails (see \cref{fig:algorithmic_primitive})---an approach reminiscent of brute-force search. We summarize this procedure in \cref{alg:state preparation}. 

\begin{algorithm}[H]
\caption{State preparation}
    \label{alg:state preparation}
    \begin{algorithmic}
    \Require SAT formula $\phi(n,m)$ with $n$ variables and $m$ clauses, angle $\theta$, tolerance $\epsilon$
    \Ensure Prepared state vector $\ket{\psi_\mathrm{out}}$
        \State \hrulefill
        \State translate $\phi(n,m) \mapsto \{C_j(\theta), P_j(\theta)\}_{j=1}^m$
        \State $r^* \gets \left\lceil \frac{\ln((\epsilon \cdot \cos^n(\theta/2))^{-1})}{\ln(\mu^{-1})} \right\rceil$ \Comment{see \cref{lem:cycle_bound}}
        \State $\ket{\psi_0} \gets \ket{+}^{\otimes n}$
        \For{$i$ in $1:r^*$} 
        \For{$j$ in $1:m$} 
        \State  $x \gets \textsc{Measure} \left( \ket{\psi_{(i-1)m + j}}, \left\{C_j(\theta),P_j(\theta)\right\}\right)$ 
        \State \Comment{single shot, $x =\pm 1$}
        \If{$x=1$} \Comment{clause check passed}
        \State $\ket{\psi_{(i-1)m + j}} \gets C_j(\theta) \ket{\psi_{(i-1)m + j-1}} /p(C_j(\theta))$
        \State \Comment{post-measurement state} 
        \Else 
        \State restart
        \EndIf
        \EndFor
        \EndFor 
        \State \Return $\ket{\psi_\mathrm{out}} \gets \ket{\psi_{r^* m}}$
    \end{algorithmic}
\end{algorithm}

To build intuition, we first discuss the case $\theta=\frac{\pi}{2}$, where we encode the Boolean values in the computational basis states. The initial state vector is the uniform superposition over all computational basis state vector 
\begin{align}\label{eqn:initial state}
    \ket{\psi_0} = \ket{+}^{\otimes n} = \frac{1}{\sqrt{2^n}} \sum_{\xx \in \{0,1\}^n} \ket{\xx}.
\end{align}
A $k$-literal clause is violated by a fraction of $\frac{1}{2^k}$ of all binary strings and satisfied by the rest. Passing the corresponding clause check projects into the space $\operatorname{im}(C_i(\frac{\pi}{2}))$ spanned by strings satisfying the clause, where else a failing means projecting into $\operatorname{im}(P_i(\frac{\pi}{2}))$, spanned by strings violating the clause. Thus, with each passed clause check, we eliminate more and more computational basis states from the superposition in \cref{eqn:initial state}. After passing all $m$ clause checks, only satisfying solution strings remain in the superposition. A single computational basis measurement then outputs a binary string representing a classical solution. Here, the exponential factor dominating the runtime is determined by the success probability of passing $m$ clause checks in sequence. This results in an unsatisfactory runtime of $\mathcal{O}(2^n)$, matching the classical brute-force search performance. By deferring measurements to the end of the circuit (see Ref.~\cite{Wilde2013}) and using amplitude amplification~\cite{Brassard2002}, the overall runtime can be brought down to $\mathcal{O}(2^{n/2})$, matching the performance of Grover's algorithm~\cite{grover1996}. The claimed runtimes are rigorously proven in \cref{thrm:run time (unrotated)}.

However, we can modulate this success probability by choosing smaller values for $\theta$. In fact, in the limit $\theta=0$, where \True{} and \False{} assignments cannot be discerned anymore, a clause check never fails. However, increasing the success probability comes at a price. For $0<\theta< \frac{\pi}{2}$, the clause projectors $P_i(\theta)$ no longer mutually commute. As a consequence, performing all $m$ clause checks successfully is not sufficient to converge to the ground space of $H(\theta)$. Instead, we have to successfully measure multiple \emph{cycles} of all $m$ clause checks, iteratively converging towards the ground space (see \cref{fig:rotated_states} [bottom]). This introduces a trade-off between the convergence rate and the success probability, with numerics in Ref.~\cite{benjamin2017} indicating a ``sweet spot'' for some non-trivial choice of $\theta$. We give an in-depth analysis of this trade-off in \cref{subsec:time complexity of state preparation}.

\subsubsection{Solution readout}
In the case $\theta=\frac{\pi}{2}$, the readout of a solution string can be performed by a single computational basis measurement. However, for $0<\theta< \frac{\pi}{2}$, the \True{}- and \False{}-assignments are no longer encoded into orthogonal states and can therefore no longer be discerned by a single measurement. In fact, the situation has been significantly complicated and requires sophisticated algorithmic techniques. In \cref{subsec:inferring_solution}, we develop such techniques and perform a rigorous analysis of their resource requirements, proving that for arbitrary $\theta \in (0,\frac{\pi}{2})$, a satisfying solution can be inferred using only $\mathcal{O}(\ln(n))$ copies of $\ket{\psi_{\mathrm{out}}}$.

\subsubsection{Resource reduction via parallelization}
The state preparation routine requires executing $m$ clause checks in each cycle. A naive implementation would perform these as $m$ sequential projective measurements. However, the total number of measurements can be substantially reduced by grouping clause checks into layers, where all measurements within such a layer mutually commute. All checks within a single layer can then be implemented simultaneously as a single, larger projective measurement. In \cref{subsec:parallelized_measurements}, we develop a systematic parallelization scheme based on this principle and provide a rigorous analysis of the resulting resource savings.

\subsubsection{Fixed-angle vs. evolving angle approach} \label{subsubsec:fixed evolving}
The algorithm can be implemented using one of two primary strategies for choosing the rotation angle $\theta$: a \emph{fixed-angle} or an \emph{evolving-angle} approach. In the fixed-angle setting, a single, constant $\theta \in (0,\frac{\pi}{2}]$ is chosen and used for the entire state preparation procedure. The algorithm then iteratively drives the quantum state towards the ground space of the corresponding Hamiltonian $H(\theta)$. Alternatively, the evolving-angle approach gradually adjusts the angle over the course of the computation, typically starting from an initial angle $\theta_{\mathrm{init}} > 0$ and scheduling it towards $\theta=\frac{\pi}{2}$. This strategy is reminiscent of an adiabatic-like evolution, aiming to guide the state towards the solution space corresponding to the unrotated encoding. While the evolving-angle scheme shows numerical promise \cite{benjamin2017}, our work focuses on providing a rigorous worst-case analysis of the fixed-angle approach, for which we can precisely characterize the trade-offs governing performance. In \cref{subsubsec:fixed_vs_evolving_discussion}, we argue that analyzing the fixed-angle algorithm is likely sufficient for inferring the exponential factor in the algorithm's runtime.

\subsubsection{Previous numerical findings}
Numerical simulations in Ref.~\cite{benjamin2017} benchmark the algorithm's performance on $3$-SAT instances with up to 34 variables, exploring both fixed- and evolving-angle strategies. In these studies, the evolving-angle approach was reported to be more efficient. A specific cubic schedule was used,
\begin{align}
    \theta(\textrm{cycle c}) = \theta_{\mathrm{init}} + \left(\frac{\pi}{2} - \theta_{\mathrm{init}} \right) \left( \frac{c}{c_Q}\right)^3,
\end{align}
with an initial angle of $\theta_{\mathrm{init}}=0.47\cdot \pi/2$ and $c_Q$ as the target number of cycles. When compared to Schöning's algorithm~\cite{Schoening1999}---a stochastic local search solver related to this approach---the quantum solver demonstrated superior average performance, with a reported runtime scaling of $(1.19)^n$. Two interesting qualitative features were also observed. First, the classical algorithm exhibited a much wider distribution of runtimes across different instances. Second, there was very little correlation between the set of instances that were difficult for the (arguably analogous) classical solver and those that were challenging for the quantum algorithm, suggesting the two approaches have different intrinsic notions of ``algorithmic hardness''.

\subsection{Structure of the worst-case analysis}
To establish a worst-case runtime for the fixed-angle algorithm, we must analyze four key quantities that multiplicatively determine its performance. These components form the structure of our analysis, as presented in the following sections:
\begin{enumerate}[leftmargin=*]
    \item \textbf{Number of cycles for state preparation ($r^*$):} We must determine the number of successful measurement cycles, $r^*$,  required to ensure that the output state vector $\ket{\psi_{\mathrm{out}}}$ is $\epsilon$-close to the ground space. This condition is met when the fidelity with the ground space is sufficiently high, i.e., $\norm{P_{\mathrm{GS}}(\theta)\ket{\psi_{\mathrm{out}}}}_2 \geq 1-\epsilon$, where $P_{\mathrm{GS}}(\theta)$ is the projector onto the ground space of $H(\theta)$. We derive a bound on the number of cycles in \cref{subsec:convergence}.
    \item \textbf{Cumulative success probability ($p_s$):} We need a lower bound on the probability of successfully executing these $r^*$ measurement cycles consecutively. This probability is given by
    \begin{align}
        p_s=\norm{\left(\prod_{i=1}^m C_i(\theta)\right)^{r^*} \ket{\psi_0}}_2^2.
    \end{align}
    We establish a bound on this quantity in \cref{subsec:success_prob}.
    \item \textbf{Number of readout repetitions ($R$):} After a state vector $\ket{\psi_{\mathrm{out}}}$ has been successfully prepared, we must determine the number of copies of $\ket{\psi_{\mathrm{out}}}$, $R$, that are needed to infer a classical solution from the quantum output with high confidence. This procedure 
    is analyzed in \cref{subsec:inferring_solution}.
    \item \textbf{Measurements per cycle ($l$):} Each cycle requires checking all $m$ clauses, which could naively be done in $m$ separate measurements. In \cref{subsec:parallelized_measurements}, we show how to parallelize measurements by grouping them into layers, yielding a much tighter upper bound on $l$, the number of measurements required per cycle.
\end{enumerate}
Throughout the paper, we will take the runtime $T$ to be the number of measurements for \cref{alg:quantum SAT solver}. With the above parameters, the runtime $T$ can be written as
\begin{align}
    T = \mathcal{O} \left( \frac{r^* \cdot R \cdot l}{(p_s)^{1/q}} \right).
\end{align}
For our default algorithm described in \cref{alg:quantum SAT solver}, we have $q=1$. By deferring measurements of the state-preparation routine to the end of the circuit and using additional auxiliary qubits
(see, e.g.,  Ref.~\mbox{\cite{Wilde2013}}) by virtue of using amplitude amplification~\mbox{\cite{Brassard2002}}, the runtime can be quadratically improved to $q=2$. It is worth noting that by using the amplitude amplitude-amplified version of the algorithm we loose the aforementioned practical advantage of a favorable sequence of measurements that might find a solution much faster than indicated by the worst-case runtime. In addition, much longer coherence times are needed to implement the amplitude-amplified version of the algorithm. 

\subsection{Time complexity of the state preparation} \label{subsec:time complexity of state preparation}
We begin with a key lemma that quantifies the overlap between the initial state and any rotated state.

\begin{lemma}[Overlap between the initial state and any rotated state]
\label{lemma:overlap_GS_plus}
Let $\ket{\Theta_{\xx}}$ be the rotated state vector associated to any length-$n$ binary string $\xx$. We find
\begin{align}
        \Abs{ \braket{\Theta_{\xx}|+^{\otimes n}}} = \cos^{n}\left( \frac{\theta}{2} \right) = \left(\frac{1 + \cos(\theta)}{2}\right)^{\frac{n}{2}}
        .\label{eq:overlap_plus}
\end{align}
\end{lemma}
\begin{proof}The proof is a direct calculation.
\end{proof}
Regarding commutation relations, 
we note the following.

\begin{lemma}[Commutation relations] \label{lem:commutation_T_P}
For all $\theta \in (0, \frac{\pi}{2})$, 
\begin{align}
    \left[\prod_{i=1}^m C_i(\theta), \calP_\mathrm{GS}(\theta) \right]=0.
\end{align}
\end{lemma}
\begin{proof}
    We have $P_i (\theta)\calP_\mathrm{GS}(\theta)=\calP_\mathrm{GS}(\theta)P_i(\theta)=0$, such that $\left[C_i(\theta),\calP_\mathrm{GS}(\theta) \right]=0$ and thus $[\prod_{i=1}^m C_i(\theta),\calP_\mathrm{GS}(\theta)]=0$.
\end{proof}

\subsubsection{Number of cycles for state preparation}
\label{subsec:convergence}
The method of alternating projections (see, e.g., Ref.~\cite{alternating_projections} for a summary of the results) addresses the question of how fast a product of projectors converges to the intersection of their respective images. 
\begin{proposition}[Alternating projections from Ref.~\cite{escalante2011alternating} (see Eq.~(3.8)), adjusted to our setting]
\label{thrm:alternating_proj}
    Let $M_1, M_2, \dots, M_\ell$ be subspaces 
    of the Hilbert space $\calH$. Denote with $P_{M_1}, P_{M_2}, \dots, P_{M_\ell}$ the orthogonal projectors that map to these subspaces. Further, denote with $\mathcal{P}_M$ the orthogonal projection onto the intersection $M\coloneqq \bigcap_{i=1}^\ell M_i$. Then, for $\ket{x} \in \calH$ and $r\geq 1$,
    \begin{align}
        \norm{\left(\left(\prod_{i=1}^\ell P_{M_i}\right)^r - \mathcal{P}_M\right) \ket{x} }_2 \leq \mu^r \norm{\ket{x}}_2
    \end{align}
    for a constant $\mu < 1$.
\end{proposition}
Translated to our setting, this gives
\begin{align}
    \norm{\left( \left(\prod_{i=1}^m C_i(\theta)\right)^{r^*} - \calP_\mathrm{GS}(\theta) \right) \ket{\psi}}_2 \leq \mu^{r^*} \norm{\ket{\psi}}_2
    \label{eq:alternating_proj}
\end{align}
with $\mu = \mu(\theta,n,m,k)$ and $ \calP_\mathrm{GS}(\theta)$ the orthogonal projector onto the ground space of $H(\theta)$. The interpretation of $\mu$ is the \emph{speed of convergence} to the simultaneous ground space. We will further comment on bounds for this quantity in \cref{subsubsec:rate of convergence} and \cref{sec:appendix:Detectability Lemma and Method of Alternating Projections}. Using this, the number of required cycles can be bounded as shown in the following lemma.

\begin{lemma}[Cycle bound]\label{lem:cycle_bound}
    When choosing
    \begin{align}
        r^* = \left\lceil \frac{\ln((\epsilon \cdot \cos^n(\theta/2))^{-1})}{\ln(\mu^{-1})} \right\rceil,
    \end{align}
    the resulting quantum state \begin{align}
    \ket{\psi_\mathrm{out}} = \frac{\left(\prod_{i=1}^m C_i(\theta) \right)^{r^*} \ket{\psi_0}}{\norm{\left(\prod_{i=1}^m C_i(\theta) \right)^{r^*} \ket{\psi_0}}}_2
    \end{align}
    satisfies
    \begin{align}
            \norm{\mathcal{P}_\mathrm{GS}(\theta) \ket{\psi_\mathrm{out}}}_2 \geq 1 - \epsilon
    \end{align} 
    for $\ket{\psi_0} = \ket{+}^{\otimes n}$.
\end{lemma}
\begin{proof}
First, we bound the number of cycles $r^*$ that are needed until
\begin{align}
    \norm{\left( \left(\prod_{i=1}^m C_i(\theta) \right)^{r^*} - \mathcal{P}_\mathrm{GS}(\theta) \right) \ket{\psi}}_2 \leq \eta \norm{\ket{\psi}}_2 
    \label{eq:condition_cyclces}
\end{align}
is observed for all quantum state vectors $\ket{\psi}$. The interpretation of this requirement is to bound the number of cycles that are needed until the product of single projectors approximates $\mathcal{P}_\mathrm{GS}(\theta)$ sufficiently well. Using the \emph{method of alternating projections} (see \cref{eq:alternating_proj}) we know that
\begin{align}
    \norm{\left( \left(\prod_{i=1}^m C_i(\theta) \right)^{r^*} - \mathcal{P}_\mathrm{GS}(\theta) \right) \ket{\psi}}_2 \leq \mu^{r^*} \norm{\ket{\psi}}_2
\end{align}
with $\mu<1$. If we choose
\begin{align}
    r^* \geq  \left\lceil \frac{\ln(\eta)}{\ln(\mu)} \right\rceil = \left\lceil \frac{\ln(\eta^{-1})}{\ln(\mu^{-1})} \right\rceil
    \label{eq:necessary_rounds}
\end{align}
we meet the condition in Eq.~(\ref{eq:condition_cyclces}). This raises the question of how to choose $\eta$ in order to ensure that the resulting quantum state vector $\ket{\psi_\mathrm{out}}$ is $\epsilon$-close to the target state. To this end, we lower bound our figure of merit by
\begin{align}
        & \norm{\mathcal{P}_\mathrm{GS}(\theta) \ket{\psi_\mathrm{out}}}_2 \\
        \nonumber
        & = \frac{\norm{\mathcal{P}_\mathrm{GS}(\theta) \left(\prod_{i=1}^m C_i(\theta) \right)^{r^*} \ket{\psi_0}}_2}{\norm{\left(\prod_{i=1}^m C_i(\theta) \right)^{r^*} \ket{\psi_0}}_2}\\
        \nonumber
        & =  \frac{\norm{\mathcal{P}_\mathrm{GS}(\theta) \ket{\psi_0}}_2}{\norm{\left(\left(\prod_{i=1}^m C_i(\theta) \right)^{r^*} - \mathcal{P}_\mathrm{GS}(\theta) + \mathcal{P}_\mathrm{GS}(\theta)\right)\ket{\psi_0}}_2}\\
        \nonumber
         & \geq \frac{\norm{\mathcal{P}_\mathrm{GS}(\theta) \ket{\psi_0}}_2}{\eta + \norm{\mathcal{P}_\mathrm{GS}(\theta) \ket{\psi_0}}_2} \\
         \nonumber
        &\geq 1 - \frac{\eta}{\norm{\mathcal{P}_\mathrm{GS}(\theta) \ket{\psi_0}}_2}.
        \nonumber
    \end{align}
    For the second equality, we have used $\left[\prod_{i=1}^m C_i(\theta), \calP_\mathrm{GS}(\theta) \right]=0$ (\cref{lem:commutation_T_P}) and for the first inequality, the method of alternating projections combined with the triangle inequality is used. Note that
    \begin{align}
        \norm{\mathcal{P}_\mathrm{GS}(\theta) \ket{\psi_0}}_2 & \geq\norm{\proj{\Theta_{\xx}}\ket{\psi_0}}_2 \\
        \nonumber
        & = \cos^n(\theta/2)
    \end{align}
    where we have used for the inequality the fact that $\proj{\Theta_{\xx}}  \subseteq \operatorname{im}(\mathcal{P}_\mathrm{GS}(\theta))$ for any solution-encoding binary string $\xx$. For the equality, the relation from \cref{lemma:overlap_GS_plus} has been  used.
    Therefore,
    \begin{align}
        \norm{\mathcal{P}_\mathrm{GS}(\theta) \ket{\psi_\mathrm{out}}}_2 \geq 1 - \frac{\eta}{\cos^n(\theta/2)}.
    \end{align}
    By choosing $\eta = \epsilon \cdot \cos^n(\theta/2)$, the requirement is met.
\end{proof}

\subsubsection{Cumulative success probability}
\label{subsec:success_prob}
\begin{lemma}[Success probability] \label{lem:success probability}
    When choosing $\ket{\psi_0} = \ket{+}^{\otimes n}$, the overall success probability is bounded by
    \begin{align}
        p_s\left(r^*, \{C_i(\theta) \}_{i=1}^m, \ket{\psi_0} \right) \geq \left(\frac{1 + \cos(\theta)}{2}\right)^n.
    \end{align}
\end{lemma}
\begin{proof}
It should be clear that $p_s\left(r,(.), (.) \right)$ is monotonically decreasing for a fixed input as $r$ grows. Therefore,
\begin{align}
    &p_s\left(r^*, \{C_i(\theta) \}_{i=1}^m, \ket{\psi_0} \right)\\
    \nonumber
    &\geq \lim_{c\rightarrow\infty} p_s\left(c, \{C_i(\theta) \}_{i=1}^m, \ket{\psi_0} \right)\\
    \nonumber
    &= \lim_{c\rightarrow\infty} \norm{\left(\prod_{i=1}^m C_i(\theta) \right)^{c} \ket{\psi_0}}_2^2\\
    \nonumber
    & = \norm{\mathcal{P}_\mathrm{GS}(\theta)\ket{\psi_0}}_2^2\\
    \nonumber
    & \geq \norm{\ket{\Theta_{\xx}} \braket{\Theta_{\xx}|\psi_0}}_2^2\\
    \nonumber
    &= \left\lvert \braket{\Theta_{\xx}|\psi_0} \right\rvert ^2
    \nonumber
\end{align}
in fact, due to Ref.~\cite{alternating_projections}. In the above, we have used that $\proj{\Theta_{\xx}}  \subseteq \operatorname{im}(\mathcal{P}_\mathrm{GS}(\theta))$
for any binary string $\xx$ that is a solution. This result is not surprising; however, it shows that the success probability ultimately depends on the overlap with the input state. Using \cref{lemma:overlap_GS_plus}, we find the desired result.
\end{proof}

\subsubsection{Rate of convergence}
\label{subsubsec:rate of convergence}
Per \cref{lem:cycle_bound}, the necessary number of cycles $r^*$ depends on the rate of convergence $\mu$. For a specific instance, this quantity can be naturally upper bounded by the spectral gap $\Delta(\theta,n,m,k)$ of the Hamiltonian
\begin{align} \label{eqn:SAT-encoding Hamiltonian}
    H(\theta) = \sum_{i=1}^m P_i(\theta).
\end{align}
To obtain a worst-case bound, we define $\Delta= \Delta(\theta,n,m,k)$ to be the smallest possible lower bound over the family of such $n$ variable $3$-SAT Hamiltonians with fixed angle $\theta$. Via the \emph{detectability lemma} \cite[Lemma 1]{Anshu_2016}, we then obtain the bound
\begin{align} \label{eqn:detectability lemma}
    \mu \leq \frac{1}{\sqrt{\Delta/g^2 + 1}} \leq \left(1 - \frac{\Delta}{4g^2}\right),
\end{align}
where $g$ denotes the maximal number of projectors that any given projector does not commute with. For more details on bounding the rate of convergence via the \emph{detectability lemma}, see \cref{app:DL_details}.
The existing lower bounds for our particular type of Hamiltonian gap are not particularly strong, and finding tighter bounds remains an open challenge. 
As an alternative formulation, we relate the quantity $\mu$ in \cref{sec:alternating_projections} to a bound on the \emph{angle of subspaces}, also known as \emph{Friedrichs angle} \cite{escalante2011alternating,badea_announcement,badea_proofs} and studied in the \emph{method of alternating projections} community. However, proving worst-case bounds on $\mu$ in our setting as a function of $(\theta,n,m,k)$ seems challenging also in this case.
The most important results surrounding both the \emph{detectability lemma} and the \emph{method of alternating projections} are highlighted in \cref{sec:appendix:Detectability Lemma and Method of Alternating Projections}.

\subsubsection{Summary}
We summarize our results so far.
\begin{theorem}[State preparation cost]
\label{thrm:runtime of state preparation} \cref{alg:state preparation} prepares a state vector $\ket{\psi_\mathrm{out}}$ that is $\epsilon$-close to the ground space, satisfying $\norm{\calP_\mathrm{GS}(\theta) \ket{\psi_\mathrm{out}}} \geq 1-\epsilon$. The runtime to achieve this with success probability at least $1-\delta$ are bounded by
\begin{align}
    T_{\mathrm{S}}(\theta,n,m,\epsilon) \leq &m \left\lceil \frac{\ln \left(\epsilon^{-1}\right) + n \ln \left( \left[\cos \left(\frac{\theta}{2} \right)\right]^{-1} \right)}{\ln\left( \mu^{-1}\right)} \right\rceil \nonumber \\ 
    &\, \times \ln\left(\delta^{-1}\right) \left( \frac{2}{1+\cos(\theta)}\right)^{n/q},
\end{align}
where $q=1$ per default and $q=2$ for the amplitude-amplified version.
\end{theorem}

\begin{proof}
    Denote with $p_s$ the probability of successfully performing all $m$ clause checks in sequence. The failure probability of trying this $\kappa$ times and remaining unsuccessful is $(1-p_s)^\kappa \leq \exp(-p_s \kappa)$. For this quantity to remain smaller then $\delta$, we thus need $\kappa \geq \ln(\delta^{-1}) p_s^{-1}$. In each of the $\kappa$ attempts, we perform at most $r^*$ cycles of $m$ clause checks. The runtime is therefore upper bounded by $T_S \leq m r^* \ln(\delta^{-1}) p_s^{-1}$. Substituting the requirement for $r^*$ (\cref{lem:cycle_bound}) and the lower bound for $p_s$ (\cref{lem:success probability}) yields the claimed upper bound. The claimed runtime for the amplitude-amplified version follows directly as a corollary. By deferring all measurements until the end of the computation (see, e.g., Ref.~\cite{Wilde2013}), amplitude amplification techniques can be applied (cf.\  Ref.~\cite{Brassard2002}). Since each execution of the non-amplitude-amplified state preparation routine succeeds with probability $((1 + \cos(\theta))/2)^n$, applying amplitude amplification yields the claimed runtime with $q=2$.
\end{proof}

\subsection{Inferring the solution}
In this section, we address the readout routine of the algorithm. To do so, the following lemma will become useful as it shows that an overlap property implies a bound on trace-norm closeness.
\begin{lemma}[Overlaps implying trace-norm closeness] \label{lem:inferring_a_solution:trace_dist_bound}
    Let us denote by $\calP_\mathrm{GS}(\theta)$ the projector onto the ground space. If $\norm{\calP_\mathrm{GS}(\theta)\ket{\psi_\mathrm{out}}}_2 \geq 1-\epsilon$,
    then there is a state vector $\ket{\psi^*} \in \operatorname{im}(\calP_\mathrm{GS}(\theta))$ such that the trace distance between the output state vector $\ket{\psi_\mathrm{out}}$ and $\ket{\psi^*}$ satisfies
    \begin{align}
        D(\ketbra{\psi^*}{\psi^*}, \ketbra{\psi_\mathrm{out}}{\psi_\mathrm{out}}) \leq \sqrt{2\epsilon}.
    \end{align}
\end{lemma}
\begin{proof}
As $\calP_\mathrm{GS}(\theta)$ is the orthogonal projector onto the ground space, we have
    \begin{align}
        \norm{\calP_\mathrm{GS}(\theta) \ket{\psi_\mathrm{out}}}_2 &= \max_{\ket{\psi} \in \operatorname{im}(\calP_\mathrm{GS}(\theta))} |\braket{\psi|\psi_\mathrm{out}}|.
    \end{align}
    Denote with $\ket{\psi^*}$ the (not necessarily 
     unique) state vector that attains this maximum. Then we have 
    \begin{align}
        D(\ketbra{\psi^*}{\psi^*}, \ketbra{\psi_\mathrm{out}}{\psi_\mathrm{out}}) &= \sqrt{1-|\braket{\psi^*|\psi_\mathrm{out}}|^2}\\
        \nonumber
        &= \sqrt{1 - \norm{\calP_\mathrm{GS}(\theta) \ket{\psi_\mathrm{out}}}_2^2} \\
        \nonumber
        &\leq \sqrt{2\epsilon}.
        \nonumber
    \end{align}
\end{proof}

We want to learn the solution state vector $\ket{\psi^*}$ from measurements of $\ket{\psi_\mathrm{out}}$. \cref{lem:inferring_a_solution:trace_dist_bound} tells us how much the measurements obtained from $\ket{\psi_\mathrm{out}}$ deviate from the ones we would have obtained had we perfectly prepared $\ket{\psi^*}$.

\label{subsec:inferring_solution}
\subsubsection{Unique solution}
The readout scheme in case of a promised unique solution is already described in Ref.~\cite{benjamin2017}, albeit without a completely rigorous resource analysis. Here, we restate their approach (\cref{alg:unique_solution_readout}) and derive a rigorous bound on $R$, taking into account that our final state is only $1-\epsilon$ close to the ground space in fidelity. A unique solution implies that $\ket{\psi^*}$ is a product state vector and the measurement outcomes on each qubit are uncorrelated, allowing us to treat them in isolation. This also holds approximately true for $\ket{\psi_\mathrm{out}}$. Now, prepare and measure the ground state repeatedly and assign each qubit to \True{} or \False{} based on a majority vote on the outcomes. We summarize this procedure in \cref{alg:unique_solution_readout} and provide rigorous performance guarantees in \cref{thrm:unique_solution_readout}.

\begin{algorithm}[H] 
\caption{Unique solution readout}
\label{alg:unique_solution_readout}
    \begin{algorithmic}
        \Require \#variables $n$, angle $\theta$, failure probability $\delta$
        \Ensure Satisfying assignment $\myvec{s}$ 
        \State \hrulefill
        \State $\bb \gets \myvec{0}_{n}$ \Comment{$\bb \in \bbZ^n$, stores measurement outcomes}
        \State  $\myvec{s} \gets \myvec{0}_n$ \Comment{$\myvec{s} \in \{0,1\}^n$, stores solution}
        \State $\epsilon \gets \frac{\left(1-\frac{1}{\sqrt{2}}\right)^2}{8} \sin^2(\theta)$ \Comment{see \cref{thrm:unique_solution_readout}}
        \State $R \gets 2\ln \left(\frac{n}{\delta}\right) \left(\ln\left(\frac{2}{1+\cos^2(\theta)}\right) \right)^{-1}$ \Comment{see \cref{thrm:unique_solution_readout}}
        \For{$i$ in $1:R$}
        \State prepare $\ket{\psi_\mathrm{out}}$ with tolerance $\epsilon$\Comment{see \cref{alg:state preparation}}
        \For{$j$ in $1:n$} 
        \State $b_j \mathrel{+}= \textsc{Measure}(\ket{\psi_\mathrm{out}},Z_j)$ \Comment{single shot}
        \EndFor 
        \EndFor
        \For{$i$ in $1:n$} \Comment{majority vote on each qubit}
        \If{$\mathrm{sgn}(b_i) = +1$} 
        \State $s_i \gets 0$ 
        \Else 
        \State $s_i \gets 1$ 
        \EndIf
        \EndFor
        \State \textbf{return} $\myvec{s}$
    \end{algorithmic}
\end{algorithm}

\begin{theorem}[Unique solution readout]
\label{thrm:unique_solution_readout}
Choosing
\begin{align}
    \epsilon = \frac{\left(1-\frac{1}{\sqrt{2}}\right)^2}{8} \sin^2(\theta),
\end{align}
\cref{alg:unique_solution_readout} succeeds with probability at least $1-\delta$ using 
\begin{align}
    R \geq \frac{2\ln \left(\frac{n}{\delta}\right)}{\ln\left(\frac{2}{1+\cos^2(\theta)}\right)}
\end{align}
copies of $\ket{\psi_\mathrm{out}}$.
\end{theorem}
\begin{proof}
    Upon measuring $\ket{\psi^*}$ in the $z$-basis, the correct value is returned with probability at least $\tilde{p}=\frac{1}{2}\left(1+\sin(\theta)\right)$. According to \cref{lem:inferring_a_solution:trace_dist_bound}, the actual state we prepare $\ket{\psi_\mathrm{out}}$ is at least $\sqrt{2\epsilon}$-close to $\ket{\psi^*}$ in trace distance, such that the correct truth value is returned with probability at least $p=\frac{1}{2}\left(1+\sin(\theta)\right)-\sqrt{2\epsilon}$. Performing a majority vote, we effectively estimate the value of $p$. \cref{alg:unique_solution_readout} fails if we observe the incorrect outcome more than half of the time. This corresponds to the estimate $\hat{p}$ of $p$ being smaller than $\frac{1}{2}$. This failure probability is upper-bounded by 
    \begin{align} \label{eqn:inferring_solution_bernoulli_concentration}
        \bbP \left[\hat{p}\leq \frac{1}{2} \right] < \exp \left(-R \cdot  \infdiv*{1/2}{p} \right),
    \end{align}
    where $\infdiv*{\alpha}{\beta}$ denotes the KL-divergence between two Bernoulli distributions with parameters $\alpha$ and $\beta$. By computing the KL-divergence as
    \begin{align}
        \infdiv*{1/2}{p} &= \frac{1}{2} \ln \left( \frac{1}{4p(1-p)} \right) \\
        &= \frac{1}{2} \ln \left( \frac{2}{1+\cos^2(\theta)} \right),\nonumber
    \end{align}
    the required bound on $R$ follows directly by substituting into \cref{eqn:inferring_solution_bernoulli_concentration}.
\end{proof}

\subsubsection{Multiple solutions}

In the above case of a promised unique solution, we were able to exploit the (approximate) tensor-product structure of $\ket{\psi_\mathrm{out}}$ to read out the satisfying assignment efficiently. This strategy is no longer possible if there are multiple satisfying assignments. Let us denote by $\calS$ the set of all solution-encoding binary strings. Then, $\ket{\psi^*}$ might be an arbitrary superposition of solution state vectors $\ket{\Theta_{\xx}}$ for $\xx \in \calS$, for which \cref{alg:unique_solution_readout} fails. Instead, we present an iterative strategy that fixes the truth values of the qubits one by one. To this end, note that upon performing a $Z$-measurement on the first qubit, the expectation value is in the interval $[-\sin(\theta), \sin(\theta)]$. Writing $\ket{\psi^*}$ as a superposition of solution state vectors,
\begin{align}
    \ket{\psi^*} = \sum_{\xx \in \calS} \alpha_{\xx} \ket{\Theta_{\xx}},
\end{align}
the outcome $\braket{Z_1}=-\sin(\theta)$ only occurs if no solution string with $\ket{\theta}$ on the first qubit contributes to the superposition, corresponding to no solution string with the first bit set to $\True{}$. Similarly, $\braket{Z_1}=\sin(\theta)$ only occurs if all contributing solution strings have $\ket{\theta}$ on the first qubit. 

With this in mind, our scheme works as follows: Starting with the first qubit, we estimate $\braket{Z_1}$ to a precision that allows us to \emph{exclude} either that $\braket{Z_1}=-\sin(\theta)$ or that $\braket{Z_1}=\sin(\theta)$. If we can guarantee that $\braket{Z_1}\neq -\sin(\theta)$, then we can safely fix the first bit to $\True{}$, as it is guaranteed that there are solutions consistent with this assignment. If, instead, we can guarantee that $\braket{Z_1}\neq\sin(\theta)$, we fix the first bit to $\False$ by analogous reasoning. We depict this argument in \cref{fig:readout_multiple_solutions}. 

\begin{figure}
    \begin{center}
    \includegraphics[width = 0.95\linewidth]{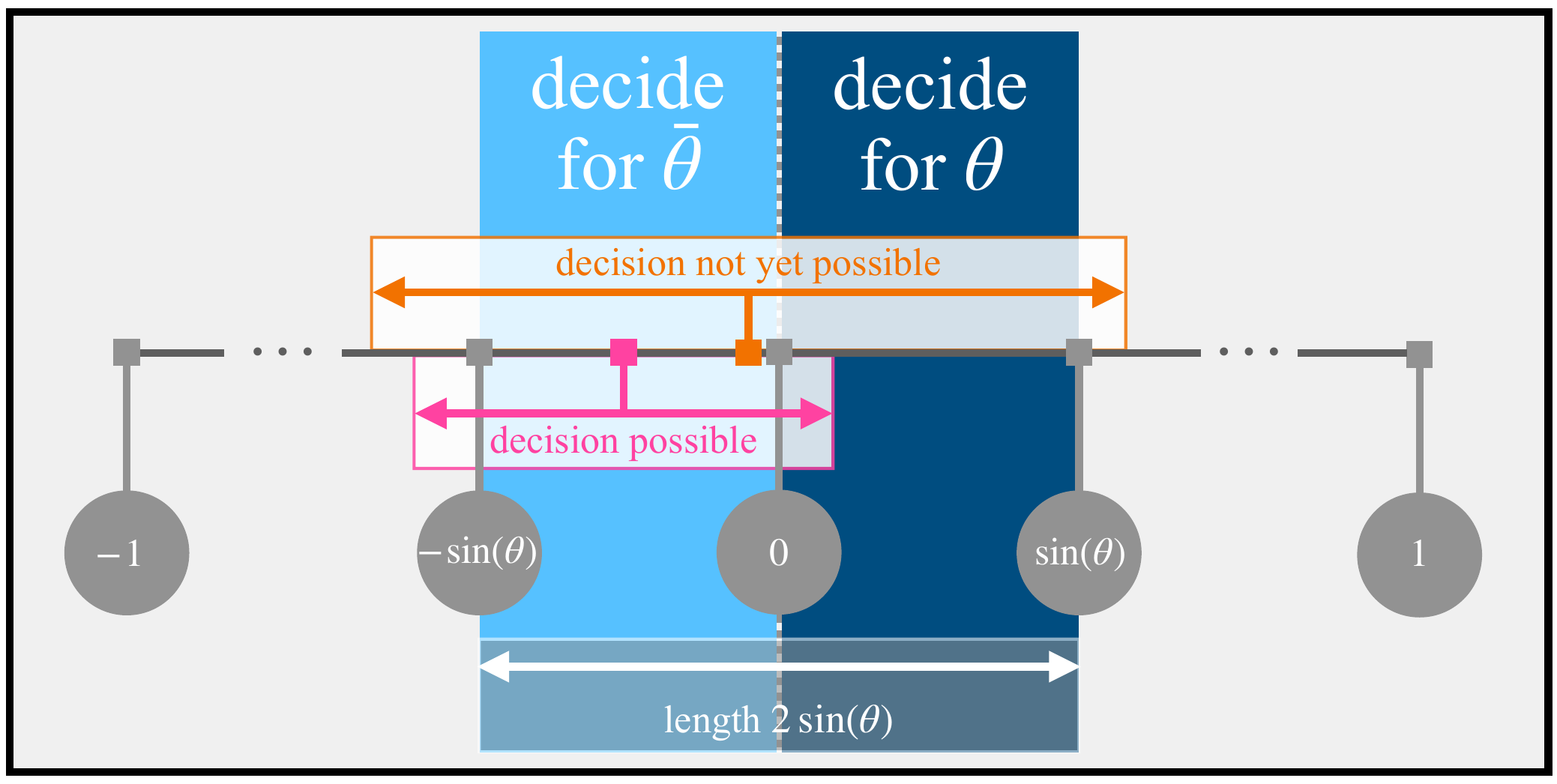}
    \end{center}
    \caption{Decision process for the proposed readout. If the estimate upon performing a $Z$-measurement on one of the qubits excludes for sure either $-\sin(\theta)$ or $\sin(\theta)$, then decide for the corresponding assignment variable. To be precise: if $-\sin(\theta)$ is excluded, then propagate \True{} on corresponding variable, if $\sin(\theta)$ is excluded, then propagate for \False{} on variable $x_i$. In the orange-colored case, we are not yet in a clear regime. Therefore, we need more shots to minimize the error bar. In the pink-colored regime, we can be relatively sure that \True{} is the right choice. Therefore, fix the variable.}
    \label{fig:readout_multiple_solutions}
\end{figure}

To \emph{hardcode} this fixed variable in the Hamiltonian, we \emph{propagate a variable} through the SAT formula by fixing its value and simplifying the resulting clauses. We illustrate this in \cref{alg:clause-propagation}.

\begin{algorithm}[H]
\caption{Propagation of a fixed variable through a SAT formula}
\label{alg:clause-propagation}
\begin{algorithmic}
\Require SAT formula $\phi$ with variables $\{x_1,\dots,x_n\}$, fixed assignment $x_j = v \in \{\True{},\False{}\}$
\Ensure Simplified SAT formula $\phi'$
\State \hrulefill
\State $\phi' \gets \emptyset$ \Comment{empty SAT formula}
\For{each clause $c_i$ in $\phi$}
    \If{$c_i$ is satisfied under $x_j = v$}
        \State remove $c_i$ from $\phi$ \Comment{satisfied clause is discarded}
    \ElsIf{$c_i$ contains literal $x_j$ that is false under $v$}
        \State remove literal $x_j$ from $c_i$ \Comment{clause is simplified}
        \State add simplified $c_i$ to $\phi'$
    \Else
        \State add $c_i$ unchanged to $\phi'$
    \EndIf
\EndFor
\State \Return $\phi'$
\end{algorithmic}
\end{algorithm}

We now repeat this decision process using the updated SAT formula to assign a truth value to the next bit.
This process is repeated for all $n$ variables until all bits have fixed truth values, thus giving a solution to the original SAT instance.
Note that the dimension of the ground space reduces by propagating choices on the assignment of certain variables further and further. The overall decision procedure is summarized in \cref{alg:multiple_solution_readout} and rigorous performance guarantees are provided in \cref{thrm:multiple_solution_readout}. In each iteration of the outer for-loop, the number of variables---and therefore the number of required qubits---is reduced by one. 

\begin{algorithm}[H]
\caption{Multiple solution readout}
\label{alg:multiple_solution_readout}
    \begin{algorithmic}
        \Require SAT formula $\phi$ with $n$ variables, angle $\theta$, failure probability $\delta$
        \Ensure Satisfying assignment $\myvec{s}$
        \State \hrulefill
        \State $\myvec{s} \gets \myvec{0}_n$ \Comment{$\myvec{s} \in \{0,1\}^n$, stores solution}
        \State $\epsilon \gets \frac{\sin^2(\theta)}{8}$ \Comment{see \cref{thrm:multiple_solution_readout}}
        \State $r \gets 8\ln\left(\frac{2n}{\delta}\right)\sin^{-2}(\theta)$ \Comment{see \cref{thrm:multiple_solution_readout}}
        \For{$i$ in $1:n$}
        \State $p \gets 0$
        \For{$j$ in $1:r$}
        \State prepare $\ket{\psi_\mathrm{out}}$ with tolerance $\epsilon$ \Comment{see \cref{alg:state preparation}}
        \State $p \mathrel{+}= \textsc{Measure} \left( \ket{\psi_\mathrm{out}},Z_j \right)$ \Comment{single shot}
        \EndFor
        \State $p \mathrel{/}= r$ \Comment{sufficiently close estimate of $\braket{Z_i}$}
        \If{$|p+\sin(\theta)| \leq |p-\sin(\theta)|$} \Comment{$\braket{Z_i}=\sin(\theta)$ ruled out}
        \State $s_i \gets 0$
        \Else \Comment{$\braket{Z_i}=-\sin(\theta)$ ruled out}
        \State $s_i \gets 1$
        \EndIf
        \State update $\phi$ according to $s_i$ \Comment{see \cref{alg:clause-propagation}}
        \EndFor
        \State \Return $\myvec{s}$
    \end{algorithmic}
\end{algorithm}

\begin{theorem}[Multiple solution readout] \label{thrm:multiple_solution_readout}
    Choosing 
    \begin{align}
        \epsilon \coloneqq \frac{\sin^2(\theta)}{8},
    \end{align}
    \cref{alg:multiple_solution_readout} succeeds with probability at least $1-\delta$ using 
    \begin{align}
        R \geq 8 \frac{n \ln \left(\frac{2n}{\delta}\right)}{\sin^2(\theta)}
    \end{align}
    copies of $\ket{\psi_\mathrm{out}}$.
\end{theorem}
\begin{proof}
    \cref{alg:multiple_solution_readout} fails if we make a truth value assignment that is inconsistent with all possible satisfying assignments of the instance. This happens if either $\braket{Z_i}=-\sin(\theta)$ and we set the $i$'th bit to $\True$ or if $\braket{Z_i}=\sin(\theta)$ and we set the $i$'th bit to $\False$. Measuring $\ket{\psi^*}$ directly, these failure modes are prevented if $\braket{Z_i}$ is estimated to precision $\epsilon'<\sin(\theta)$, such that we can rule out either $\braket{Z_i}=-\sin(\theta)$ or $\braket{Z_i}=\sin(\theta)$. Practically, we can only infer information about $\ket{\psi^*}$ from measurement of $\ket{\psi_\mathrm{out}}$, which is $\sqrt{2 \epsilon} =\frac{1}{2}\sin(\theta)$-close in trace distance (\cref{lem:inferring_a_solution:trace_dist_bound}). Thus, to rule out at least one boundary of the interval $[-\sin(\theta), \sin(\theta)]$, we need to measure $\braket{\psi_{\mathrm{out}}|Z_i|\psi_{\mathrm{out}}}$
    to precision $\epsilon''<\frac{1}{2} \sin(\theta)$, such that $\epsilon' < \sin(\theta)$ for the estimate of $\braket{Z_i}=\braket{\psi^*|Z_i|\psi^*}$ is ensured. We apply Hoeffding's inequality to obtain precision $\epsilon''$ with failure probability at most $\delta/ n$ 
    using 
    \begin{align}
        r \geq 8 \frac{ \ln \left(\frac{2n}{\delta}\right)}{\sin^2(\theta)}
    \end{align}
    rounds. We repeat this process for all $n$ qubits, applying the union bound over the individual failure probabilities. This yields an overall failure probability of at most $\delta$ after $R=nr$ rounds.
\end{proof}

Note that the multiple-solution readout matches the asymptotic performance of the unique-solution readout in the rotation angle in the limit of small angles $\theta$, i.e., we have $R=\calO(1/\theta^2)$ in both cases.
In each iteration of the outer loop of \cref{alg:multiple_solution_readout}, the Hamiltonian encoding the underlying SAT formula is updated. Let us denote by $H_j(\theta)$ the Hamiltonian at the $j$'th step, i.e., after fixing the $j$'th variable. With the Hamiltonian gap controlling the convergence rate, one might be worried that this complicates the analysis of this quantity. We show that it suffices the analyze the so-called \emph{uniform spectral gap}, as it lower bounds the gap of any updated Hamiltonian (\cref{theorem:uniform_gap_is_a_lower_bound}).
\begin{definition}[Uniform spectral gap]
    For a given Hamiltonian $H =\sum_{i \in [m]} \Pi_i$ with $\Pi_i$ being projectors, we denote by $\Delta(H)$ its spectral gap. Following Ref.~\cite{Gilyen2017}, the \emph{uniform spectral gap} of $H$ is defined as
\begin{align}
    \Delta_{\mathrm{uni}}(H) \coloneq \min_{\calI \subseteq [m]} \Delta\left(\sum_{i\in \calI} \Pi_i \right).
\end{align}
\end{definition}
\begin{theorem}[Uniform spectral gap as a lower bound]
\label{theorem:uniform_gap_is_a_lower_bound}
    The spectral gap of the Hamiltonian $H_j(\theta)$ at the $j$'th step of \cref{alg:multiple_solution_readout} is lower bounded by the uniform gap $\Delta_{\mathrm{uni}}(H(\theta,n,m,k))$ where $H(\theta,n,m,k)$ is the Hamiltonian associated with the SAT-instance.
\end{theorem}
\begin{proof}
    After the $j$'th step, $j$ variables have been assigned fixed boolean truth values.
    In the corresponding updates of the Hamiltonian (see \cref{alg:clause-propagation}), a clause is either removed from the clause set (since it is already satisfied by the hardcoded variables) or the locality of the clause decreases. In either case, the Hamiltonian after the $j$'th step can be written as
    \begin{align}
        H_j(\theta,n,m,k) = \sum_{i \in \calI_j} P^{(j)}_i(\theta)
    \end{align}
    with $\calI_j\subseteq [m]$ and $P^{(j)}_i(\theta)$ being a projector acting non trivially on at most $k$ qubits. We show $H_j(\theta) \geq H_{j-1}(\theta)$, which directly implies the desired claim. Each projector $P_i^{(j)}$ is either equal to $P_i^{(j-1)}$ or can be derived from it by deleting a literal of the corresponding clause. In the former case, $P_i^{(j)}(\theta) \geq P_i^{(j-1)}(\theta)$ holds trivially. For the latter case, assume w.l.o.g. that the fixed variable corresponds to the first qubit. Then we have
    \begin{align}
        P^{(j-1)}_i(\theta) &= \proj{x} \otimes \proj{\phi_{\mathrm{rest}}} \text{with $x \in \{\theta^\perp, \bar{\theta}^\perp\}$} ,\\
        P^{(j)}_i(\theta) &= \Id \otimes \proj{\phi_{\mathrm{rest}}},
    \end{align}
    where 
    $\proj{\phi_{\mathrm{rest}}}$ denotes the part of the projector not acting on the first qubit. It holds that $P^{(j)}_i(\theta) \geq P_i^{(j-1)}(\theta)$ because $\Id - \ketbra{x}{x} \geq 0$ and $\proj{\phi_{\mathrm{rest}}} \geq 0$ such that $P_j(\theta)-P_{j-1}(\theta) = (\Id-\ketbra{x}{x}) \otimes \proj{\phi_{\mathrm{rest}}} \geq 0$. Repeated application directly gives 
    \begin{align}
        \sum_{i\in \calI_j} P_i^{(j)}(\theta) \geq \sum_{i\in \calI_j} P_i^{(0)}(\theta).
    \end{align}
    By the definition of the uniform spectral gap of the original Hamiltonian, we directly see that the spectral gap of $H_j(\theta,n,m,k)$ is lower bounded by $\Delta_{\mathrm{uni}}(H(\theta,n,m,k))$, which concludes the proof.
\end{proof}
\subsection{Parallelizing measurements}
\label{subsec:parallelized_measurements}
In this section, we show how to group the $m$ clause checks, each being $k$-local, into at most $\ell(n,k) \leq \frac{\sqrt{k}}{\sqrt{2 \pi}}(2e)^k \ln(n)$ layers, where all clause checks within one such layer commute, thus allowing simultaneous measurement. To achieve this, we utilize \emph{perfect hash families} as a technical tool. 

\begin{definition}[Perfect hash family \cite{mehlhorn2013data}] \label{def:perfect hash family} Consider a collection of $N$ functions $f_i:[n] \mapsto [k]$, where $[n]=\{1,2,3,\dots,n\}$ and $[k]$ analogously. We call $\{f_i\}_{i=1}^N$ a $(N,n,k)$-perfect hash family if, for every $S \subset [n]$ with $|S|=k$, there is a function $f_i$ such that $f_i:S \to [k]$ is one-to-one.
\end{definition}
Our definition of a perfect hash family slightly deviates from the standard. To simplify notation, we fix $|S|=k$ instead of allowing $|S| \leq k$. Perfect has families can equivalently be defined as an array with certain properties. While less intuitive, this is often used in constructing such families.

\begin{definition}[Perfect hash families, array characterization] \label{def:perfect hash family (array)}
    A perfect hash family is uniquely characterized by an array of size $N \times n$ on $k$ symbols such that for any $N \times k$ subarray, there is at least one row comprised of distinct symbols.
\end{definition}
To construct a perfect hash family, we resort to a simple, deterministic scheme, the \emph{density algorithm} \cite{colbourn2008perfect_hash_families_deterministically}. The density algorithm constructs a set of rows $R_{\#}$ that form an array satisfying \cref{def:perfect hash family (array)}. In this context, we define the notion of a \emph{partial row} $\rr=(r_1, \dots, r_n)$, where $r_i \in [k] \cup \{\star\}$ with $\star$ denoting a value that has not yet been fixed. When a partial row is updated by replacing an unfixed value $r_i=\star$ with $r_i=x$, where $x \in [k]$, we denote the resulting row with $\sigma(\rr,x,i)$. The density algorithm starts with an empty set $R_{\#}$ and iteratively adds rows to this set. Each row starts completely undetermined as $\rr=(\star, \star, \dots, \star)$ and is updated entry by entry, maximizing a cost function. To this end, 
define
\begin{align}
    \delta(\rr) \coloneqq \sum_{\substack{\calI \subseteq [n], |\calI| = k}} \lambda(\calI) \cdot\chi(\calI,\rr),
\end{align}
where we sum over all subsets $\calI \subseteq [n]$ of size $k$. The function $\lambda:\calI \to \{0,1\}$ evaluates to zero if there is a row in $R_{\#}$ that has distinct (fixed) values on all entries indexed by $\calI$. Otherwise, $\lambda(\calI) = 1$. For $\chi(\calI,\rr)$, consider the restriction of $\rr$ to the index set $\calI$ and denote with $s$ the number of entries with fixed values. We have $\chi(\calI,\rr)=\frac{(k-s)(k-s-1)\dots 1}{k^{k-s}}$ if all entries of $\rr$ indexed by $\calI$ are distinct, otherwise $\chi(\calI,\rr)=0$. The next undetermined entry of $\rr$ is then set to
\begin{align}
    r_i = \Argmax_{x \in [k]}  \delta(\sigma(\rr,x,i)).
\end{align}
Once all $\star$-values in $\rr$ are replaced, the row is added to the set $R_{\#}$. This process is repeated until the terminating condition
\begin{align}
    \sum_{\calI \in [n], |\calI|=k} \lambda(\calI) = 0
\end{align}
is met. We summarize this procedure in \cref{alg:density_algorithm}. The density algorithm runs in time $\mathcal{O}(\ln(n) n^k)$ (Lemma 2.3 from \cite{colbourn2008perfect_hash_families_deterministically}), which is polynomial for any fixed $k$. The size $N$ of the resulting hash family grows only logarithmically in the number of variables $n$ (\cref{lem:density algorithm}). 

\begin{algorithm}[H]
\caption{Density algorithm \cite{colbourn2008perfect_hash_families_deterministically}}
\label{alg:density_algorithm}
    \begin{algorithmic}
        \Require $n$, $k$
        \Ensure $(\mathcal{O}(\ln(n)), n, k)$-perfect hash family $R_{\#}$
        \State \hrulefill
        \State $R_{\#} \gets \{\}$ \Comment{empty set of rows}
        \While{$\sum_{\calI \in [n], |\calI|=k} \lambda(\calI) \neq 0$}
            \State $\rr \gets (\star, \star, \dots, \star)$ \Comment{size of $\rr$ is $n$}
            \For{$i$ in $1:n$}
            \State $r_i = \Argmax_{x \in [k]}  \delta(\chi(\rr,x,i))$
            \EndFor
            \State $R_{\#} \gets R_{\#} \cup \rr$
        \EndWhile
        \State \Return $R_{\#}$ \Comment{array characterization, see \cref{def:perfect hash family (array)}}
    \end{algorithmic}
\end{algorithm}

\begin{lemma}[Lemma 2.2 from Ref.~\cite{colbourn2008perfect_hash_families_deterministically}, adjusted] \label{lem:density algorithm}
    \cref{alg:density_algorithm} constructs an $(N,n,k)$-perfect hash family with $N \leq c_k \ln \binom{n}{k} \leq k c_k \ln(n)$ where 
    \begin{equation}
        c_k = \left(\ln \left( \frac{k^k}{k^k-k!} \right) \right)^{-1} < \frac{e^k}{\sqrt{2\pi k}}.
    \end{equation}
\end{lemma}
\begin{proof}
    Follow the original argument and set $v = t$. For the upper bound on $c_k$, we have 
    \begin{align}
        c_k = \left(-\ln \left(1-\frac{k!}{k^k} \right)\right)^{-1} <  \frac{k^k}{k!}
    \end{align}
    as $\ln(1-x) < -x$ for $x \in (0,1)$. The desired claim then follows from the lower bound $\sqrt{2 \pi k} \left( \frac{k}{e} \right)^k < k!$.
\end{proof}

\begin{lemma}[Lemma 2.3 from Ref.~\cite{colbourn2008perfect_hash_families_deterministically}, corrected] \label{lem:density algorithm (runtime)}
    The density algorithm runs in time 
    \begin{align}
        \mathcal{O}\left( k^{3/2} \ln(n) (en)^k\right).
    \end{align}
\end{lemma}
\begin{proof}
    The original proof mistakenly claims $c_k<2$, when in fact $c_k < \frac{e^k}{\sqrt{2\pi k}}$, see \cref{lem:density algorithm}. Following the original argument with $v=t$ (which corresponds to $k$ in our setting) and the correct bound on $c_k$, we obtain the claimed runtime.
\end{proof}

Using a perfect hash family, we can now group the clause checks into projective measurements. Let each clause check, represented by projector $P_i(\theta)$, be identified by an $n$-symbol string, mapping $\bar{\theta}^\perp \to 0$, $\theta^\perp \to 1$ and $\Id \to \mathrm{I}$. As an example, for $n=6$, the clause check defined by
\begin{align}
    P_i(\theta) = \proj{\theta^\perp_2 \theta^\perp_3 \bar{\theta}^\perp_5} \otimes \Id_{[6]\setminus \{2,3,5\}}
\end{align}
would be identified with the string $I11I0I$. We say two strings are compatible if they match in each position in which both strings do not have an $I$ as an entry. For instance, $11I0$ is compatible with $1I00$, but not with $1I01$. Note that if two strings $\xx, \yy \in \{0,1,\mathrm{I}\}^n$ are both compatible with a binary string $\bb \in \{0,1\}^n$, then $\xx$ and $\yy$ are also compatible with each other. Thus, we can characterize a group of commuting measurements with a length-$n$ binary string. We refer to such a group as a \emph{layer}.

We now use a perfect hash family to construct a family of binary strings and show that each clause check is compatible with at least one such layer.
\begin{theorem}[Parallelizing measurements] \label{thrm:number_of_layers}
    All clause checks can be performed using 
    \begin{equation}
    \ell(n,k) \leq \frac{\sqrt{k}}{\sqrt{2 \pi}}(2e)^k  \ln(n)
    \end{equation}
    projective measurements. 
\end{theorem}
\begin{proof}
    All clause checks in a layer commute and can thus be realized in a single projective measurement. We give an explicit construction of the layers. First, construct a $(\ell, n, k)$ perfect hash family using the density algorithm. Note that we can derive a layer from a function $f_i$ by identifying each element in the image of $f_i$ with a value of either $0$ or $1$. This results in a length-$n$ binary string. For every $f_i$, we construct all $2^k$ possible layers, resulting in $2^k\ell$ layers in total. Now, by virtue of $\{f_i\}$ being a perfect hash family, for any clause check, there is at least one $f_i$ that maps the indices of the non-trivial tensor factors of the projector $P_i(\theta)$ to $k$ distinct numbers. As a consequence, the string $\xx \in \{0,1,\mathrm{I}\}^n$ associated with the clause check is compatible with one of the $2^k$ binary strings (layers) derived from $f_i$. From \cref{lem:density algorithm}, we obtain a number of $2^k \ell \leq 2^k k c_k \ln(n) \leq \frac{\sqrt{k}}{\sqrt{2 \pi}}(2e)^k  \ln(n)$ layers. 
\end{proof}
By the above theorem, we have now grouped the clause checks into $\ell(n,k)=\mathcal{O}(\sqrt{k} (2e)^k \ln(n))$ projective measurements, each corresponding to a layer. In practice, each of these measurements can be performed by measuring all the clause checks that have been grouped into the corresponding layer. Thus, each layer is associated with a two-outcome measurement asking whether the state is in the subspace prohibited by the clauses sorted into the layer or not. In terms of the mutually commuting clause checks $C_1(\theta), \dots ,C_\ell(\theta)$ sorted into a given layer, the measurement can be written as $\{\prod_{i=1}^\ell C_i(\theta), \Id-\prod_{i=1}^\ell C_i(\theta)\}$. As all clause checks commute, this measurement can be implemented in the circuit model using a multi-controlled-NOT gate. Using auxiliary qubits, this can be done in the circuit model in depth $\mathcal{O}(\ln(n))$ \cite{he2017decompositions}.

While the perfect hash family construction guarantees an upper bound on the number of required measurement layers, this is a worst-case analysis. In practice, specific instances may permit a more compact grouping of commuting clause checks.

\subsection{Overall time complexity}
At this stage, we find that the expected runtime of this quantum algorithm can be bounded as follows:
\begin{theorem}[Overall time complexity]\label{thrm:overall runtime}
For every satisfiable $3$-SAT instance \cref{alg:quantum SAT solver} finds a solution in time
\label{thrm:summary_runtimes}
    \begin{align} \label{eqn:overall runtime 1}
         T(\theta, n, \delta) &= \mathcal{O} \left( \ln(n) \frac{\left[\ln \left(\frac{1}{1-\cos^2(\theta)}\right) + n \ln \left( \frac{1}{\cos \left(\frac{\theta}{2} \right)} \right)\right]}{\ln\left( \mu^{-1}\right)}  \right. \nonumber \\ 
    &\quad \left. \times \ln\left(\delta^{-1}\right) \left( \frac{2}{1+\cos(\theta)}\right)^{n/q} t(\theta, n,\delta) \right),
    \end{align}
    where $q=1$ per default and $q=2$ for the amplitude-amplified version.
    Here, $\delta$ denotes the joint failure probability of state preparation and readout algorithm. 
    The term $t(\theta,n,\delta)$ is determined by the readout algorithm used. 
    For $3$-SAT instances with possibly multiple satisfying assignments, \cref{alg:multiple_solution_readout} applies and we have 
    \begin{align} \label{eqn:overall runtime 2}
        t_{\mathrm{multiple}}(\theta, n, \delta) &= \mathcal{O} \left( \frac{\ln\left( \frac{n}{\delta}\right)}{\sin^2(\theta)}\right).
    \end{align}
    If promised a unique solution, \cref{alg:unique_solution_readout} applies, and the readout cost reduces to
    \begin{align}\label{eqn:overall runtime 3}
        t_{\mathrm{unique}}(\theta, n, \delta) &= \mathcal{O} \left( \frac{\ln \left(\frac{n}{\delta}\right)}{\ln \left(\frac{2}{1+\cos^2(\theta)}\right)} \right).
    \end{align}
\end{theorem}
\begin{proof}
    The overall time complexity is 
    \begin{align} \label{eqn:proof summary runtime}
        T(\theta, n,m,\delta) = T_{\mathrm{S}}(\theta,n,m,\epsilon) \cdot t(\theta,n,\delta),
    \end{align}
    where $T_{\mathrm{S}}(\theta,n,m,\epsilon)$ denotes the time complexity of the state preparation routine (\cref{thrm:runtime of state preparation}) and the overhead $t(\theta, n,\delta)$ from reading out a solution. The $m$ clause checks can be implemented in $\mathcal{O}(\ln(n))$ measurements (\cref{thrm:number_of_layers}). The readout procedure dictates the necessary precision $\epsilon$ for the state preparation. Both the unique solution readout and the multiple solution readout require precision polynomial in $1-\cos^2(\theta)$ (see \cref{thrm:unique_solution_readout} and \cref{thrm:multiple_solution_readout}), which gives rise to the same asymptotic behavior $\ln(\epsilon^{-1})=\mathcal{O}(\ln(1/[1-\cos^2(\theta)]))$. Substituting $m$ and $\epsilon$ in \cref{eqn:proof summary runtime} directly gives \cref{eqn:overall runtime 1}. The overhead $t_{\mathrm{unique}}(\theta,n,\delta)$ from executing \cref{alg:unique_solution_readout} corresponds exactly to the number of copies of $\ket{\psi_\mathrm{out}}$ that are required per \cref{thrm:unique_solution_readout}, as expected. However, the overhead $t_{\mathrm{multiple}}(\theta,n,\delta)$ from executing \cref{alg:multiple_solution_readout} corresponds to the cost of reading out a single qubit, not the overall number of copies of $\ket{\psi_\mathrm{out}}$ needed. This is because in each iteration of \cref{alg:multiple_solution_readout}, the number of qubits necessary for the state preparation routine decreases by one. The overall runtime is thus given by
    \begin{align}
        T(\theta,n,m,\delta) &=\sum_{i=1}^n T_{\mathrm{S}}(\theta,i,m,\epsilon) \cdot t_{\mathrm{multiple}}(\theta,n,\delta)\\
        \nonumber
        &= \mathcal{O}(T_{\mathrm{S}}(\theta,n,m,\epsilon) \cdot t_{\mathrm{multiple}}(\theta,n,\delta)),
    \end{align}
    explaining the missing factor of $n$ compared to \cref{thrm:multiple_solution_readout}.
\end{proof}
For the sake of clarity, we assume $k=3$ in the above theorem. For general $k$-SAT, we pick up an additional factor of $\sqrt{k}(2e)^k$ from the parallelization scheme for the measurements (see \cref{subsec:parallelized_measurements}).
\begin{remark}[Termination]
    If \cref{alg:quantum SAT solver} does not terminate within the allocated time bound from \cref{thrm:overall runtime}, we output \texttt{UNSAT}. This output is correct with probability at least $1-\delta$.
\end{remark}

\section{Analysis of some restricted input classes} \label{sec:Analysis of some restricted input classes}
In this section, we analyze the asymptotic runtime for some restricted input classes. Crucially, we show that by choosing $\theta=\theta(n)$ correctly, we can exponentially improve the runtime on certain inputs compared to a fixed $\theta=\pi/2$.
We start by deriving a sharp upper bound on the runtime for $\theta=\pi/2$. Here, all check operations $C_i(\theta=\pi/2)$ commute, such that \cref{alg:state preparation} converges to the ground space in a single round of successful measurements. The runtime is thus determined by the probability of performing $m$ successful measurements in sequence.
\begin{theorem}[Overall time complexity in the unrotated setting]
\label{thrm:run time (unrotated)}
    For fixed angle $\theta=\pi/2$, \cref{alg:quantum SAT solver} finds a satisfying solution for a SAT instance with probability at least $(1-\delta)$ in time
    \begin{align} \label{eqn:run time (unrotated) 1}
        T(n,m,\delta) &=  \mathcal{O} \left(m \ln \left(\delta^{-1} \right) \left(\frac{2^n}{d_{\mathrm{sol}}}\right)^{1/q} \right),
    \end{align}
    where $q=1$ per default and $q=2$ for the amplitude-amplified version. In the above expression, $d_{\mathrm{sol}}$ denotes the dimension of the ground space, i.e., the number of satisfying solutions. The bound is asymptotically optimal in $\delta$ and $2^n/d_{\mathrm{sol}}$, i.e.,
    \begin{align} \label{eqn:run time (unrotated) 2}
        T(n,\delta) = \Omega \left(\ln \left(\delta^{-1} \right) \left(\frac{2^n}{d_{\mathrm{sol}}}\right)^{1/q} \right),
    \end{align}
\end{theorem}
\begin{proof}
    All check operations $C_i = C_i(\pi/2)$ commute, such that \cref{alg:state preparation} converges to the ground space in a single round of successful measurements since $\prod_{i=1}^m C_i = \calP_{\mathrm{GS}}$. A satisfying assignment can be read out in a single computational basis measurement. The runtime is therefore determined by the probability $p_s$ of performing $m$ successful measurements in sequence. As all projectors commute, measuring them one after the other is equivalent to performing the two-outcome measurement asking whether the state is in the ground space or not, i.e., the measurement $\{\prod_{i=1}^m C_i, \Id - \prod_{i=1}^m C_i \}$. As we start in the state vector $\ket{+}^{\otimes n}$, which is the equal superposition state over all basis states, this measurement succeeds with $p_s=d_{\mathrm{sol}}/2^n$. The failure probability of no successful trial after $\kappa$ tries is given as $(1-p_s)^\kappa$. To guarantee a failure probability of at most $\delta$, we require $(1-p_s)^\kappa \leq \delta$ such that $\kappa \geq \ln(\delta^{-1}) / (-\ln(1-p_s)) \geq \ln(\delta^{-1})/p_s$ trials are sufficient. Each trial performs at most $m$ clause checks, which concludes the proof of \cref{eqn:run time (unrotated) 1}. By definition, the bound $\kappa \geq \ln(\delta^{-1}) / (-\ln(1-p_s))$ is sharp, such that using $-\ln(1-p_s)=\Theta(p_s)$, we obtain $\kappa = \Theta(\ln(\delta^{-1}) /p_s)$. The number of trials $\kappa$ strictly lower bounds $T$, which proves \cref{eqn:run time (unrotated) 2}. The amplitude-amplified versions are a direct corollary which follows from deferring measurements to the end and using amplitude amplification~\cite{Brassard2002}.
\end{proof}
The worst-case is attained for Unique-SAT instances. By definition, these instances have at most one satisfying solution ($d_{\mathrm{sol}}=1$), leading to a runtime of $\mathcal{O}(2^n)$ for the non-amplitude-amplified version and $\mathcal{O}(2^{n/2})$ for the amplitude-amplified version. Note that this scaling behavior is reproduced by a randomized version of classical brute-force search where the next candidate solution is chosen uniformly at random over all length-$n$ bit strings.

\begin{definition}[Unate-SAT instances] A SAT formula is called unate if each variable $b_i$ either only appears as a positive literal $b_i$ or only as a negative literal $\bar{b}_i$ \cite{balogh2023nearlyksatfunctionsunate}. This is equivalent to all clause checks $C_i(\theta)$ of a given instance mutually commuting for all $\theta \in (0,\pi/2]$.
\end{definition}

Unate-SAT instances are classically trivially solvable in polynomial time. If a variable appears only as a positive literal, assign $\True$; if it only appears as a negative literal, assign $\False$. This gives rise to an $\mathcal{O}(n)$-time algorithm. However, Unate-SAT instances are an interesting testing ground for the measurement-driven quantum SAT solver. It is easy to construct Unate-SAT instances with a uniquely satisfying solution. Furthermore, all clause checks commute such that the algorithm converges in a single cycle, implying a constant Hamiltonian gap.

\begin{corollary}[Unate-SAT separation] \label{cor:unate-SAT separation}
    By setting the rotation angle $\theta$ such that $\cos(\theta) = 1-\frac{2}{n}$, \cref{alg:quantum SAT solver} finds a satisfying assignment for a Unate-SAT instance in time
    \begin{align}\label{eqn:unate-SAT separation 1}
        T(n) = \mathcal{O}\left(\ln^2(n) n \right).
    \end{align}
    This contrasts sharply with the fixed angle case $\theta=\pi/2$, where the worst-case runtime is lower bounded by 
    \begin{align}\label{eqn:unate-SAT separation 2}
        T(n) = \Omega(2^{n/q})
    \end{align}
    where $q=1$ per default and $q=2$ by using the amplitude-amplified algorithm.
\end{corollary}
\begin{proof}
    For a Unate-SAT instance, all clause projectors commute. Therefore, the state preparation routine \cref{alg:state preparation} converges in a single measurement cycle. By this, the overall runtime expression in \cref{thrm:overall runtime} simplifies to 
    \begin{align}
        T(\theta,n,\delta) = \mathcal{O}\left(\ln(n) \ln(\delta^{-1}) \left(\frac{2}{1+\cos(\theta)} \right)^{n/q} t(\theta,n,\delta) \right).
    \end{align}
    The upper bound in \cref{eqn:unate-SAT separation 1} then follows by inserting $\cos(\theta)=1-\frac{2}{n}$ into the above expression. Unate-SAT instances can have multiple satisfying assignments, such that \cref{eqn:overall runtime 2} applies. Substituting $\theta$, we obtain $t(n,\delta)=\mathcal{O}(\ln(n) n)$. For the lower bound, we construct a Unate-SAT instance with a uniquely satisfying solution ($d_{\mathrm{sol}}=1$) by simply picking a one-literal clause for every variable. \cref{eqn:unate-SAT separation 2} then follows directly from \cref{thrm:run time (unrotated)}.
\end{proof}
For clarity, we omit the dependence on $\ln(\delta^{-1})$ in \cref{cor:unate-SAT separation}. By setting our rotation angle appropriately, we manage to match the optimal classical runtime up to $\log$-factors.

\section{Open questions, comments, and future work}

\subsection{Towards improving the lower bound on the Hamiltonian gap}
Proving a better lower bound on the Hamiltonian gap $\Delta$ directly translates to a better guarantee on the convergence rate $\mu$ and therefore an improved runtime guarantee for the overall algorithm. As established by the \emph{detectability lemma} (see \cref{eqn:detectability lemma} and \cref{sec:appendix:Detectability Lemma and Method of Alternating Projections}), the runtime scales with the gap as
\begin{align}
    \mathcal{O} \left( \frac{1}{\ln\left(\mu^{-1}\right)}\right) = \mathcal{O} \left( \frac{1}{\Delta}\right).
\end{align}
Looking at the overall runtime in \cref{thrm:overall runtime}, we can make strong arguments about the \emph{form} this gap must take. The Hamiltonian gap must vanish exponentially fast as $n$ increases. If this were not the case, we could tune $\theta$ such that the exponential dependence arising from the success probability in \cref{thrm:overall runtime} becomes arbitrarily small. As a consequence, this would lead to a sub-exponential algorithm for $k$-SAT ($k\geq 3$), which is widely believed to be false under standard complexity-theoretic assumptions (ETH). We formally prove this claim in \cref{prop:form of hamiltonian gap}.

This argument establishes that the general form of the bound we seek is known. The central challenge, therefore, is to find the tightest possible base $\beta$ for this exponential scaling. Doing so directly translates into determining the dominating exponential factor in the algorithm's worst-case runtime.

To formalize this, let $\calH = \calH(\theta, n,m,k)$ be the family of Hamiltonians $H(\theta)=\sum_{i=1}^m P_i(\theta)$ that correspond to a satisfiable $k$-SAT instance with $n$ variables and $m$ clauses. We define the \emph{worst-case spectral gap} $\Delta(\calH)$ for this class as 
\begin{align}
    \Delta(\calH) \coloneqq \min_{H \in \calH} \Delta(H),
\end{align}
where $\Delta(H)$ denotes the spectral gap of $H$.

\begin{theorem}[Form of the Hamiltonian gap]\label{prop:form of hamiltonian gap}
    Assuming the \emph{exponential time hypothesis} (ETH), there is a constant $\beta = \beta(\theta,k) > 1$ such that
    \begin{align}
        \Delta(\calH(\theta,n,m,k)) = \Omega(\beta^{-n}).
    \end{align}
\end{theorem}
\begin{proof}
    Assume for contradiction that the gap $\Delta(n)$ vanishes sub-exponentially in $n$, meaning $1/\Delta(n)$ grows sub-exponentially. From \cref{thrm:overall runtime} and the \emph{detectability lemma} (\cref{eqn:detectability lemma}), the algorithm's runtime $T(n)$ scales as
    \begin{align}
        T(n) \propto \frac{1}{\Delta(n)} \left(\frac{2}{1+\cos(\theta)}\right)^{n/q},
    \end{align}
    where $q=1$ per default and $q=2$ for the amplitude-amplified version. In the above, we omitted polynomial terms.
    If $1/\Delta(n)$ grows sub-exponentially, the runtime's exponential base is $2/(1+\cos(\theta))$, which can be brought arbitrarily close to 1 by choosing $\theta$ sufficiently close to 0. This would imply a sub-exponential runtime for $k$-SAT, violating the ETH. 
\end{proof}

We formulate the precise challenge of finding a non-trivial lower bound on the worst-case spectral gap as an open problem below.
\begin{oproblem}[Worst-case Hamiltonian gap scaling]
    Establish a non-trivial lower bound on the worst-case spectral gap $\Delta$. That is, find any constant $\beta=\beta(\theta,k)$ with $1<\beta<2$ such that $\Delta=\Omega(\beta^{-n})$ and characterize the $k$- and $\theta$-dependence of this scaling.
    A complete solution would determine the tightest possible bound by identifying the smallest value for $\beta$ that satisfies this relation, thereby characterizing the precise asymptotic scaling.
\end{oproblem}
\begin{remark}[Bounding worst-case Friedrichs angles]
    We remark that this open problem could equivalently be solved by bounding the worst-case Friedrichs angle between the respective images of the layers as a function of $(\theta,n,m,k)$ (see \cref{sec:alternating_projections} for the framework).
\end{remark}

\begin{remark}[Runtime for $2$-SAT]
    So far, we have not been able to prove that the algorithm’s runtime collapses from exponential to polynomial if restricting the input class to $2$-SAT. Proving such a statement would be of value in its on right. On the other hand, it is also possible that the runtime remains exponential. In this respect, it is worth noting that Grover’s algorithm for $2$-SAT also has an exponential runtime as it does not make use of any locality structure (see also \cref{subsec:local_structures}).
\end{remark}

In \cref{sec:worst-case_bounds_for_gaps}, we review known worst-case bounds on gap sizes and explain why these approaches are not helpful in our setting, but also point towards some perspectives.

\subsection{Towards exploiting local structures quantumly}
\label{subsec:local_structures}
The prospect of achieving a super-quadratic quantum advantage for $3$-SAT hinges on whether an algorithm can exploit the problem's local structure in a uniquely quantum way. Solving a general, unstructured SAT instance is believed to be computationally as hard as black-box search, a conjecture formalized by the SETH, which posits that SAT can only be solved in time $\calO(2^n)$. For black-box search, 
the quadratic quantum speedup achieved by Grover's algorithm is optimal \cite{Bennett_1997}, which gives rise to the QSETH, the conjecture that SAT cannot be solved faster than in time $\calO(2^{n/2})$, even on a quantum computer. However, the core assumption underpinning QSETH---the absence of exploitable structure--- does not hold for 
$3$-SAT. Its $3$-local nature provides a structural foothold that the best classical algorithms leverage to significantly outperform brute-force search. To date, the best proven worst-case bounds for $3$-SAT are achieved by simply applying Grover-style amplitude amplification to these advanced classical solvers. In this hybrid paradigm, the crucial task of exploiting local structure is offloaded to the classical component. The quantum contribution remains generic, which limits these approaches to a quadratic speedup compared to the best classical solvers. This inherent limitation frames a central challenge for quantum combinatorial optimization: Can a quantum computer natively leverage the locality structure that separates 3-SAT from unstructured search?

We end this section by phrasing a precise condition on when
the rotated algorithm beats brute-force search. As the algorithm is compatible with amplitude-amplification, this would immediately give a potential mechanism to boost an algorithmic primitive super-quadratically by combining amplitude-amplification with a rotated encoding.
\begin{remark}
    Assuming the form of the gap dictated by the ETH (see \cref{prop:form of hamiltonian gap}), there is a constant $\beta = \beta(\theta,k)>1$ such that 
    \begin{align}
        \Delta(\calH(\theta,n,m,k)) = \Omega(\beta^{-n}).
    \end{align}
    To beat Grover search, the dominating exponential factor of \cref{thrm:overall runtime} has to be smaller than $2^{n/2}$, i.e. there has to be a $\theta$ such that,
    \begin{align}
        \left(\frac{2\cdot \beta^2}{1+\cos(\theta)}\right)^{n/2} \leq 2^{n/2}\\
        \Rightarrow \beta^2\leq (1 +\cos(\theta)).
    \end{align}
\end{remark}

\subsection{Towards analyzing the average-case behavior of the algorithm}
As we have seen in the intermediate steps of the proof from \cref{theorem:gap_scaling_benjamin} that addresses the scaling of the spectral gap, the Hamming distance between strings is involved. In this section, we make the reader aware of the following fact: For large $n$, the Hamming distance $D_{\xx \yy}$ behaves like a binomial distribution, i.e., $D_{\xx \yy} \sim \operatorname{Binom}(n, 1/2)$. Therefore,
\begin{align}
    \mathbb{E}_{\xx, \yy} \left[\cos^{D_{\xx \yy}}(\theta) \right]& \approx \sum_{\ell=0}^n \binom{n}{\ell} \left(\frac{1}{2}\right)^n \cos^\ell(\theta)\\
    & = \left(\frac{ 1 + \cos(\theta)}{2}\right)^n.\nonumber
\end{align}
Note that this expression coincides with the state preparation's success probability (c.f.~\cref{lem:success probability}).

This observation suggests that the average-case behavior of the algorithm is closely tied to probabilistic properties of Hamming distances between random strings. A natural open direction is to investigate whether this approximation can be formalized into a rigorous average-case analysis, and more generally, to determine how the distributional structure of Hamming distances influences the algorithm's performance. This fact might also serve as a starting point for a possible explanation of why certain instances appear ``algorithmically easy'' for this algorithm while being ``algorithmically hard'' for Schöning's algorithm.

\subsection{Understanding and improving the algorithm}
Below, we outline several future directions for improving and extending the algorithm.
\subsubsection{Fixed- vs. evolving-angle algorithm}
\label{subsubsec:fixed_vs_evolving_discussion}
An important design choice in our algorithm is whether to use a fixed
measurement angle~$\theta$ (see \cref{alg:quantum SAT solver}) or to let the angle $\theta$ evolve over time (see \cref{subsubsec:fixed evolving}). 
In the fixed-angle setting, the algorithm faces a clear trade-off in performance: As we rigorously established in \cref{thrm:run time (unrotated)}, setting $\theta=\pi/2$ leads to worst-case scaling of $\mathcal{O}(2^n)$ for the non-amplitude-amplified version, matching classical brute-force search, and $\mathcal{O}(2^{n/2})$ for the amplitude-amplified version, matching  Grover search \cite{grover1996}. Conversely, a smaller angle $\theta$ boosts the success probability, but numerical evidence from Ref.~\cite{benjamin2017} suggests that as $\theta \to 0$, the convergence rate slows down significantly, thus increasing the overall runtime. This trade-off suggests the existence of a non-trivial ``sweet spot'' for $\theta$, corresponding to optimized performance. By contrast, an evolving-angle-schedule is naturally reminiscent
of adiabatic approaches, where the gradual adjustment of parameters can help
guide the system through the energy landscape. It is natural to ask how these two approaches compare.

From a complexity-theoretic perspective, we conjecture that the fixed-angle approach, when choosing $\theta$ to be the optimal value, matches the runtime of the evolving angle approach up to polynomial factors. This conjecture rests on two observations. First, our analysis of the readout routines (\cref{subsec:inferring_solution}) shows that, barring the state preparation sub-routine, they scale polynomially in both $n$ and $\theta$. Therefore, the readout overhead incurred in a fixed-angle setting does not affect the dominant exponential scaling. Second, while an evolving schedule might be able to improve convergence, we suspect that finding a desirable $\theta$-schedule is itself a computationally hard problem, analogous to finding optimal paths in adiabatic quantum computing. We suspect that this difficulty prevents improving the exponential scaling compared to an optimally chosen fixed $\theta$.
Clarifying this would also be of broader interest to adiabatic quantum computing.

\subsubsection{Local resampling strategies}
The runtime of the algorithm could likely be significantly improved by replacing the current \emph{global resampling strategy}---which discards the entire state and restarts from scratch whenever an unfavorable measurement occurs---with a more efficient \emph{local resampling strategy}. In this approach, only the $k$ qubits involved in the undesired outcome are resampled, while the rest of the system is left untouched. This partial resampling could, e.g., be done by resetting the local part to the state vector $\ket{+}^{\otimes k}$ or by resampling one of the ``forbidden'' variable assignments, reminiscent of the approach used in Schöning's algorithm~\cite{Schoening1999}. However, finding meaningful and helpful analytical expressions for this scenario seems challenging (as already pointed out in Ref.~\cite{cubitt2023dissipative}).

\subsubsection{Extension to MAX-SAT}
The way we presented the algorithm so far, it is only applicable to SAT. A natural question to ask is whether a slight modification of the algorithm can be made to extend it to solve the maximum satisfiability problem (MAX-SAT). For MAX-SAT, the problem-encoding Hamiltonian is no longer guaranteed to be frustration-free. As such, the key algorithmic challenge is that the projective measurements onto local constraint spaces that are currently used can easily drive the state out of the global optimum. As suggested in Ref.~\cite{cubitt2023dissipative}, this difficulty may be mitigated by replacing \emph{projective measurements} with \emph{weak measurements}. These are generalized quantum measurements in which one operator is close to the identity. Such measurements gently bias the state towards satisfying local terms without fully collapsing it, thereby avoiding large deviations from the global structure. Iterating this process could, in principle, nudge the state incrementally towards an approximate global solution, providing a potential pathway to extend the algorithm from SAT to MAX-SAT. After convergence, which is guaranteed due to Ref.~\cite{cubitt2023dissipative}, our proposed readout routine from \cref{subsec:inferring_solution} could then be used to extract a solution.

\subsubsection{Hardware implementations and noise-robustness}
We start by pointing out that two experimentally relevant challenges arise when considering smaller angles. First, smaller values of $\theta$ demand longer coherence times and more auxiliary qubits, as they require more successive measurement cycles for the state preparation routine (see \cref{alg:state preparation}). Second, the readout routine (see \cref{alg:multiple_solution_readout}) requires more precision as $\theta$ decreases.

However, as discussed in Refs.~\cite{cubitt2023dissipative, bombin2021}, the proposed non-amplitude-amplified, i.e., measurement-driven, quantum algorithm is still particularly well-suited for implementation on photonic quantum hardware. In particular, for photonic hardware, the proposed parallelization of measurements is straightforward. Notably, the noise and error resilience demonstrated in Ref.~\cite{cubitt2023dissipative} also applies directly to our framework. Specifically, the self-correcting nature of the algorithm ensures robustness against noise and errors (up to an error rate below a certain threshold) without incurring additional computational overhead. Moreover, we highlight a practical distinction of the measurement-driven algorithm from fixed-runtime solvers, such as Grover's algorithm or our amplitude-amplified version: this algorithm's runtime is stochastic, determined by the specific order and success of the measurements. Therefore, a favorable sequence of measurements can potentially find a solution much faster than indicated by the worst-case bound.

In contrast, the amplitude-amplified version yields a better asymptotic runtime. However, this comes at the expense of much higher coherence times and more auxiliary qubits. As such, it is less hardware-friendly on near-term devices.

\section{Conclusions}
In this work, we have provided a rigorous, worst-case runtime analysis for the measurement-driven quantum SAT solver introduced in Ref.~\cite{benjamin2017} and have pointed out that an amplitude-amplified version thereof further boosts the performance. Our analysis formally establishes the algorithm's runtime dependence on two key properties: the spectral gap of the associated Hamiltonian and the success probability of the measurements. We have demonstrated that these properties are linked by an exponential trade-off, which can be systematically controlled by the algorithm's rotation angle.

On the algorithmic side, we have significantly broadened the algorithm's practicality. We have introduced a new, rigorous readout routine capable of efficiently finding a solution even for general instances with multiple satisfying assignments. Furthermore, we have developed a measurement parallelization scheme based on perfect hash families, which groups the $m$ clause checks into $\mathcal{O}(\ln(n))$ commuting layers, each implementable as a single measurement.

We have then demonstrated the practical utility of our analytical framework. By appropriately tuning the angle $\theta$ according to our analysis, we have shown that the algorithm's runtime on Unate-SAT instances can be exponentially improved, collapsing from $\Omega(2^n)$ at $\theta=\pi/2$ to a polynomial runtime.

The algorithmic primitive examined in this work (see \cref{fig:algorithmic_primitive}) implements a brute-force search procedure. In the unrotated case, applying amplitude amplification yields a runtime of $\mathcal{O}(2^{n/2})$, matching the performance of Grover’s original algorithm~\cite{grover1996}. By introducing a rotation parameterized by an angle $0 < \theta < \pi/2$, we rigorously show that the effectiveness of this rotated-basis approach depends crucially on how the spectral gap scales with $\theta$ and $n$. This observation compared with promising numerics suggests that, for sufficiently small $k$, the method may offer a super-quadratic improvement over classical brute-force search.

This work highlights the spectral gap as the central quantity determining the algorithm's performance. The most critical open question remains establishing a non-trivial lower bound on this worst-case spectral gap for $k$-SAT. Resolving this would clarify the algorithm's ultimate potential. Moreover, our methods open avenues for future average-case analysis, extensions to MAX-SAT, and are 
likely beneficial for the analysis of other dissipation-driven algorithms.

\section*{Acknowledgements}
The authors would like to thank Paul K.~Faehrmann, Jonas Haferkamp, Carsten Schubert, Nathan Walk, Daniel Miller, and Philipp Schmoll for their fruitful discussions and valuable insights throughout various stages of the project.
The authors utilized Grammarly (Grammarly 
Inc., 2025) for grammar and spelling corrections, and the tools ChatGPT GPT-4 and GPT-5 (OpenAI, 2025) and Gemini 2.5 (Google DeepMind, 2025) were employed to refine sentence clarity. The tools were used only for language editing. All intellectual contributions, scientific reasoning, and conclusions are solely those of the authors. This work has been  supported by the BMFTR (DAQC, MUNIQC-Atoms, QuSol, 
HYBRID++, PasQuops), the Munich Quantum Valley (K-4 and K-8), the DFG (SPP 2514),
the Quantum Flagship (PasQuans2, Millenion), QuantERA (HQCC), the Clusters of Excellence MATH+ and ML4Q, the DFG (CRC183), the Einstein Foundation (Einstein Research Unit on Quantum Devices), Berlin Quantum, and the ERC (DebuQC). 

%

\clearpage

\appendix
\onecolumngrid
\savegeometry{main_layout} 
\newgeometry{verbose,tmargin=2cm,bmargin=2cm,lmargin=3cm,rmargin=3cm} 

\section{Method of alternating projections and detectability lemma} \label{sec:appendix:Detectability Lemma and Method of Alternating Projections}

Clearly, the quantity for the speed of convergence, i.e., $\mu$, is something that we cannot access in the lab. However, the same holds for a Hamiltonian gap. Therefore, we need to bound these quantities. As pointed out above, by the ETH~\cite{impagliazzo2001eth1,impagliazzo2001eth2}, it is 
unlikely to solve SAT efficiently on 
a quantum computer in the worst case. Therefore, we expect 
\begin{align}
    1/\ln(\mu^{-1}) \leq \alpha(\theta)^{-n} \textrm{ for $0<\alpha(\theta)<1$},
\end{align}
c.f. \cref{thrm:overall runtime}. 
Ref.~\cite{alternating_projections} gives a bound for the speed of convergence $\mu$ and a recipe for explicit calculation. Moreover, in the language of physics, this quantity can be related to the gap via the \emph{detectability lemma} (see Ref.~\cite{Anshu_2016}) and the quantum union bound (see Ref.~\cite{Gao_2015}). The details are pointed out below.
\subsection{Method of alternating projections approach}
\label{sec:alternating_projections}
A substantial subfield in linear algebra, known as the \emph{method of alternating projections}, has been devoted to analyzing the speed of convergence of iterative projections onto subspaces. A cornerstone of this theory is the \emph{Friedrichs angle}, which defines the angle between two subspaces and plays a central role in characterizing convergence behavior \cite{friedrichs1937certain}. A comprehensive survey of this subfield and its key results can be found in Ref.~\cite{deutsch1995angle}.

In this work, we focus specifically on the contributions from Refs.~\cite{badea_announcement, badea_proofs}, which extend the notion of the \emph{Friedrichs angle} to the setting of \emph{multiple subspaces}. They define the angle between several subspaces in the following way:
\begin{proposition}[Generalization of the angle of subspaces to multiple subspaces \cite{badea_announcement, badea_proofs}]
\label{thrm:angle_of_subspaces}
For $M_{1},\dots, M_\ell$ being closed subspaces (with $\ell\geq 2$) let $M\coloneqq \cap_{i=1}^\ell M_i$ be the intersection of all $M_i$. The \emph{Friedrichs angle}, which can be associated with these subspaces, is defined as
\begin{align} 
c(M_1, \dots, M_\ell) &= \sup \left\{\frac{2}{\ell-1} \frac{\sum_{j<k}\operatorname{Re} \braket{x_j\vert x_k}}{\sum_{i=1}^\ell \lvert \lvert x_i \rvert \rvert^2}: x_j \in M_j \cap M^\perp,  \sum_{i} \lvert \lvert x_i \rvert \rvert^2 \neq 0 \right\}\\
&= \sup \left\{\frac{1}{\ell-1} \frac{\sum_{j\neq k} \braket{x_j\vert x_k}}{\sum_{i=1}^\ell\braket{x_i\vert x_i}}: x_j \in M_j \cap M^\perp,  \sum_{i} \lvert \lvert x_i \rvert \rvert^2 \neq 0 \right\}.
\nonumber
\end{align}
\end{proposition}

\begin{proposition}[Speed of convergence \cite{badea_announcement, badea_proofs}]
\label{thrm:rate of convergence bound}
For $M_{1},\dots, M_\ell$ being closed subspaces (with $\ell \geq 2$) let $M\coloneqq \cap_{i=1}^\ell M_i$ be the intersection of all $M_i$ and $\mathcal{T} = P_{M_\ell} \dots P_{M_1}$. Suppose that $c\coloneqq c(M_1, \dots M_\ell)<1$. Then, there is quick uniform convergence of the powers $\mathcal{T}^r$ to $P_M$. More precisely,
\begin{align}
\norm{\mathcal{T}^r - P_M}_2 \leq \left(1 - \left( \frac{1-c}{4\ell}\right)^2 \right)^{r/2}.
\end{align}
\end{proposition}

We have $\norm{ \mathcal{T}(\theta)^r - P_M}_2 \leq \mu^r$. Thus, 
\begin{align}
    \mu \leq \left(1 - \left( \frac{1-c}{4\ell}\right)^2 \right)^{1/2}.
\end{align}
For our runtime estimate, we find $\ln(\mu^{-1})$, which in turn can be bounded as
\begin{align}
\ln(\mu^{-1}) \geq \ln\left(\frac{1}{ \sqrt{1 - \left(\frac{1-c}{4\ell} \right)^2}}\right) \geq  \frac{1}{2} \left(\frac{1-c}{4\ell} \right)^2.
\end{align}
Here we have used that $\ln(1-z) \leq -z$ for $z \in [0,1)$.

Deriving the desired bound from \cref{thrm:rate of convergence bound} involves two steps. In the first step, we show that all clause checks on $n$ qubits can be performed using $\mathcal{O}(\ln(n))$ projective measurements. For this purpose, the perfect hash family construction in Section~\ref{subsec:parallelized_measurements} is used. In the second step, the angle of subspaces from \cref{thrm:angle_of_subspaces} has to be determined for a given set of parameters $(\theta,n,m,k)$.

\begin{remark}[Explicit worst-case bounds]
    We are convinced that the \emph{method of alternating projections} and the associated \emph{Friedrichs angle} are precisely the right language to talk about the problem at hand. However, calculating explicit non-trivial worst-case bounds in this framework seems challenging.
\end{remark}

\subsection{Detectability lemma approach}
\label{app:DL_details}
By translating Refs.~\cite{Anshu_2016,Gao_2015} to the alternating-projections language, we find the following relation of the gap and the speed of convergence:
\begin{proposition}[Detectability lemma]
\label{prop:DL}
    Assume that $H(\theta) = \sum_{i=1}^m P_i(\theta)$ where each $P_i(\theta)$ is a projector that does not commute with at most $g$ other terms, i.e., $P_j(\theta)$ with $j\neq i$.
    The \emph{detectability lemma} from Ref.~\cite{Anshu_2016} implies the following relation between the spectral gap $\Delta(\theta,n,m,k)$ and $\mu$ from \cref{thrm:alternating_proj}:
    \begin{align}
        \mu \leq  \frac{1}{\sqrt{\Delta(\theta,n,m,k)/g^2 + 1}} \leq \left(1 - \frac{\Delta(\theta,n,m,k)}{4g^2}\right).
    \end{align}
\end{proposition}
On the other hand, the quantum union bound due to Ref.~\cite{Gao_2015} implies
the following.

\begin{proposition}[Quantum union bound/converse detectability lemma]
\label{prop:QUB}
    The \emph{quantum union bound} from Ref.~\cite{Gao_2015} implies the following relation between the spectral gap $\Delta(\theta,n,m,k)$ and $\mu$ from \cref{thrm:alternating_proj}:
    \begin{align}
        1 - 4\Delta(\theta,n,m,k) \leq \mu.
    \end{align}
\end{proposition}

Obviously, for this approach to work, the gap $\Delta(\theta,n,m,k)$ and the uniform gap $\Delta_{\mathrm{uni}}(\theta,n,m,k)$, depending on which Hamiltonian is used to run the algorithm, have to be bounded sufficiently well. It is a simple corollary of \cref{thrm:summary_runtimes} to let the worst-case runtime depend on the spectral gap of the rotated Hamiltonian that is used for running the algorithm.
\begin{theorem}[Overall time complexity as a function of the uniform gap]
\cref{alg:quantum SAT solver} finds a satisfying solution for a $3$-SAT instance in time
\label{thrm:summary_runtimes_gap}
    \begin{align}
         T(\theta, n, \delta) &= \mathcal{O} \left( \ln(n) \frac{\left[\ln \left(\frac{1}{1-\cos^2(\theta)}\right) + n \ln \left( \frac{1}{\cos \left(\frac{\theta}{2} \right)} \right)\right]}{\frac{\Delta_{\mathrm{uni}(\theta,n,m,k)}}{4g^2}} \ln\left(\delta^{-1}\right) \left( \frac{2}{1+\cos(\theta)}\right)^{n/q} t(\theta, n,\delta) \right),
    \end{align}
    where $q=1$ per default and $q=2$ for the amplitude-amplified version. In the above expression, $\delta$ denotes the joint failure probability of state preparation and readout algorithm, $g$ is an upper bound such that each $P_i(\theta)$ does not commute at most with $g$ others. We denote by $\Delta_{\mathrm{uni}}(\theta,n,m,k)$ the uniform spectral gap. The term $t(\theta,n,\delta)$ is determined by the readout algorithm used.
    For general $3$-SAT instances with possibly multiple satisfying assignments, \cref{alg:multiple_solution_readout} applies and we have 
    \begin{align}
        t_{\mathrm{multiple}}(\theta, n, \delta) &= \mathcal{O} \left( \frac{\ln\left( \frac{n}{\delta}\right)}{\sin^2(\theta)}\right).
    \end{align}
    If promised a unique solution, \cref{alg:unique_solution_readout} applies, and the readout cost reduces to
    \begin{align}
        t_{\mathrm{unique}}(\theta, n, \delta) &= \mathcal{O} \left( \frac{\ln \left(\frac{n}{\delta}\right)}{\ln \left(\frac{2}{1+\cos^2(\theta)}\right)} \right).
    \end{align}
\end{theorem}

\begin{proof}
    Starting from the runtime bounds derived in \cref{thrm:summary_runtimes} and using the \emph{detectability lemma} from \cref{prop:DL} to relate $\mu$ with the spectral gap, i.e.,
    \begin{align}
        \mu \leq 1 - \frac{\Delta_{\mathrm{uni}}(\theta,n,m,k)}{4g^2}.
    \end{align}
    We find that
    \begin{align}
        \ln(\mu^{-1}) \geq 
        -\ln\left(1 - \frac{\Delta_{\mathrm{uni}}(\theta,n,m,k)} {4g^2}\right) \geq \frac{\Delta_{\mathrm{uni}}(\theta,n,m,k)} {4g^2}.
    \end{align}
    Here we have used that $\ln(1-z) \leq -z$ for $z \in [0,1)$ which applies in our case since $\Delta_{\mathrm{uni}}(\theta,n,m,k)/g^2<1$ as noted in Ref.~\cite{Anshu_2016}. Plugging this value into the bounds from \cref{thrm:summary_runtimes}, we find the desired result. Note that in the general case, the uniform spectral gap is needed because while running the algorithm, the Hamiltonian changes (see \cref{theorem:uniform_gap_is_a_lower_bound} for the argument why the uniform gap can be used as a lower bound).
\end{proof}

While Ref.~\cite[Appendix G]{benjamin2017} has already proven a lower bound for the spectral gap, this turned out to be not tight enough to yield a speed-up over brute force search, but for completeness we state it here (in a slightly generalized form for general $k$):
\begin{theorem}[Spectral gap lower bound generalized to $k$-SAT (from Ref.~\cite{benjamin2017})]
\label{theorem:gap_scaling_benjamin}
    Let $\phi$ be a satisfiable $k$-SAT instance with $n$ variables and $m$ clauses, and let us denote by $\theta \in (0,\pi/2)$ a rotation angle. At such an angle, we associate the Hamiltonian $H(\theta)$ as defined in \cref{subsec:encoding_hamiltonian}. The spectral gap of this Hamiltonian is lower bounded by the following value
    \begin{align}
        \Delta(H(\theta)) \geq \sin^{2k}(\theta)\cdot \left(\frac{1 - \cos(\theta)}{1 + \cos(\theta)}\right)^{n}.
    \end{align}
\end{theorem}
\begin{proof}
    Although we follow the proof presented in Ref.~\cite[Appendix G]{benjamin2017}, we decided to spell it out explicitly here for the convenience of the reader. 
    Let us start by noting the  matrix inequality 
    \begin{align}
    \label{eq:trace_ineq}
       \operatorname{Tr}[A]\cdot \lambda_{\mathrm{min}}(B)\leq  \operatorname{Tr}[AB] \leq \operatorname{Tr}[A]\cdot \lambda_{\mathrm{max}}(B)
    \end{align}
    with $\lambda_{\mathrm{min}}(B)$ being the minimum eigenvalue of $B$.
    Let us choose as a non-orthogonal basis
    \begin{align}
        \left\{\ket{\Theta_{\xx}} \big \vert \xx \in \{0,1\}^n \right\}.
    \end{align}
    With this choice, we make the following ansatz for the first excited state vector
    \begin{align}
        \ket{\Phi} = \sum_{\xx \in \{0,1\}^n} \alpha_{\xx} \ket{\Theta_{\xx}}.
    \end{align}
    Since the Hamiltonian is frustration-free, i.e., has energy zero, the gap is equal to the energy eigenvalue of the first excited state, i.e.,
    \begin{align}
        \Delta(H(\theta)) & = E_1(\theta)\\
        \nonumber
      & =   \min_{\ket{\Phi}: \calP_\mathrm{GS}(\theta) \ket{\Phi} = 0, \norm{\ket{\Phi}}=1} \braket{\Phi|H(\theta)|\Phi}\\
      \nonumber
     & = \min_{\ket{\Phi}: \calP_\mathrm{GS}(\theta) \ket{\Phi} = 0, \norm{\ket{\Phi}}=1} \sum_i \braket{\Phi|P_i(\theta)|\Phi} \\
      \nonumber
     & = \min_{\ket{\Phi}: \calP_\mathrm{GS}(\theta) \ket{\Phi} = 0, \norm{\ket{\Phi}}=1} \sum_i \sum_{\xx,\yy \in \bar{\calS}_i} \alpha_{\xx}^* \alpha_{\yy} \sin(\theta)^{2k} \braket{\Theta_{\xx} \vert \Theta_{\yy}}\\
      \nonumber
     & = \sin^{2k}(\theta) \min_{\ket{\Phi}: \calP_\mathrm{GS}(\theta) \ket{\Phi} = 0, \norm{\ket{\Phi}}=1} \sum_i \sum_{\xx,\yy \in \bar{\calS}_i} \alpha_{\xx}^* \alpha_{\yy}  \braket{\Theta_{\xx} \vert \Theta_{\yy}}.
      \nonumber
    \end{align}
    Here, we have noted that $\braket{\Theta_{\xx} |P_i(\theta)|\Theta_{\yy}}$ 
    is either equal to zero or to $\sin(\theta)^{2k}\cdot \braket{\Theta_{\xx} \vert \Theta_{\yy}}$. The latter only occurs when the $k$ bits of $\xx,\yy$ on which the projector $P_i(\theta)$ acts non-trivially are equal and are equal to the exact $k$-bit string that violates the $i$'th clause. Therefore, we only sum $x,y$ over the strings $\bar{\calS}_i$ that violate the $i$'th clause.

    Let us note further that $\braket{\Theta_{\xx} \vert \Theta_{\yy}}$ can be re-expressed as $\cos^{D_{\xx \yy}}(\theta)$ where $D_{\xx \yy}$ is the Hamming distance of $\xx$ and $\yy$. This fact establishes that we can rewrite
    \begin{align}
        \braket{\Theta_{\xx} \vert \Theta_{\yy}} = \cos^{D_{\xx \yy}}(\theta) = \braket{\xx \vert M^{\otimes n} \vert \yy}
    \end{align}
    with $\ket{\xx}$ and $\ket{\yy}$ being quantum states in the computational basis associated to bitsrings $\xx$ and $\yy$ and
    \begin{align}
        M = \begin{pmatrix}
            1 & \cos(\theta)\\
            \cos(\theta) & 1
        \end{pmatrix}.
    \end{align}
    This allows us to re-express 
    \begin{align}
        \Delta(H(\theta)) &= \sin^{2k}(\theta) \min_{\ket{\Phi}: \calP_\mathrm{GS}(\theta) \ket{\Phi} = 0, \norm{\ket{\Phi}}=1} \sum_i \sum_{\xx,\yy \in \bar{\cal{S}}_i} \alpha_{\xx}^* \alpha_{\yy}  \braket{\xx \vert M^{\otimes n} \vert \yy}\\
        \nonumber
        & = \sin^{2k}(\theta) \min_{\ket{\Phi}: \calP_\mathrm{GS}(\theta) \ket{\Phi} = 0, \norm{\ket{\Phi}}=1} \operatorname{Tr}\left[\left(\sum_i  \sum_{\xx,\yy \in \bar{\calS}_i} \alpha_{\xx}^* \alpha_{\yy}  \ket{\yy}\bra{\xx}\right) M^{\otimes n}\right].
        \nonumber
    \end{align}
    This expression can be bounded from below by using the 
    matrix inequality from \cref{eq:trace_ineq}
    \begin{align}
        \Delta(H(\theta)) & \geq \sin^{2k}(\theta) \min_{\ket{\Phi}: \calP_\mathrm{GS}(\theta) \ket{\Phi} = 0, \norm{\ket{\Phi}}=1} \left\{ \lambda_{\mathrm{min}}(M)^n \cdot \operatorname{Tr}\left[\left(\sum_i  \sum_{\xx,\yy \in \bar{\calS}_i} \alpha_{\xx}^* \alpha_{\yy}  \ket{\yy}\bra{\xx}\right)\right]\right\}\\
        \nonumber
        & = \sin^{2k}(\theta) \min_{\ket{\Phi}: \calP_\mathrm{GS}(\theta) \ket{\Phi} = 0, \norm{\ket{\Phi}}=1} \left\{ \lambda_{\mathrm{min}}(M)^n \cdot \left(\sum_i  \sum_{\xx,\yy \in \bar{\calS}_i} \Abs{\alpha_{\xx}}^2  \right)\right\}.
        \nonumber
    \end{align}
    We note that the matrix $M$ has eigenvalues $(1\pm \cos(\theta))$. Moreover, each $x$ which is not a solution must be contained in at least one set $\bar{\calS}_i$. Thus, we have
    \begin{align}
        \Delta(H(\theta)) & = \sin^{2k}(\theta) (1 - \cos(\theta))^n \min_{\ket{\Phi}: \calP_\mathrm{GS}(\theta) \ket{\Phi} = 0, \norm{\ket{\Phi}}=1} \left\{ \sum_{\xx,\yy \notin \calS} \Abs{\alpha_{\xx}}^2 \right\}
    \end{align}
    where $\calS$ denotes the set of binary strings that are solutions. To bound the term $\sum_{\xx \notin \calS} \Abs{\alpha_{\xx}}^2$, we make use of the following fact: Let us partition the ansatz state into states which correspond to solutions (i.e., $\xx \in \calS$) and others which are not solutions (i.e., $\xx \notin \calS$), i.e.,
    \begin{align}
        \ket{\Phi} = \sum_{\xx} \alpha_{\xx} \ket{\theta_{\xx}} = \sum_{\xx \in \calS} \alpha_{\xx} \ket{\theta_{\xx}} + \sum_{\xx \notin \calS} \alpha_{\xx} \ket{\theta_{\xx}}.
    \end{align}
    We require that $\ket{\Phi}$ is orthogonal to all solutions state vector  $\ket{\theta_{\myvec{s}}}$ with $\myvec{s} \in \calS$. Requiring this yields the fact
    that
    \begin{align}
        \label{eq:gap_ansatz_orthogonality}
        &0  = \braket{\theta_{\myvec{s}} \vert \Phi} = \sum_{\xx \in \calS} \alpha_{\xx} \braket{\theta_{\myvec{s}} \vert \theta_{\xx}} + \sum_{\xx \notin \calS} \alpha_{\xx} \braket{\theta_{\myvec{s}} \vert \theta_{\xx}} \\
        \nonumber
        &\Leftrightarrow \sum_{\xx \notin \calS} \alpha_{\xx} \braket{\theta_{\myvec{s}} \vert \theta_{\xx}} = - \sum_{\xx \in \calS} \alpha_{\xx} \braket{\theta_{\myvec{s}} \vert \theta_{\xx}} ,
        \nonumber
    \end{align}
    which has to hold for every $\myvec{s} \in \calS$, therefore
    \begin{align}
        \label{eq:gap_ansatz_orthogonality2}
        & \sum_{\xx \in \calS} \sum_{\yy \in \calS} \alpha_{\xx}^* \alpha_{\yy} \braket{\theta_{\xx} \vert \theta_{\yy}} = - \sum_{\xx \in \calS} \sum_{\yy \notin \calS} \alpha_{\xx}^* \alpha_{\yy} \braket{\theta_{\xx} \vert \theta_{\yy}}.
    \end{align}
    Combined with the normalization we enforce, this yields
    \begin{align}
        1 = \braket{\Phi \vert \Phi} = \sum_{\xx \notin \calS} \sum_{\yy \notin \calS} \alpha_{\xx}^* \alpha_{\yy} \braket{\theta_{\xx} \vert \theta_{\yy}} - \sum_{\uu \in \calS} \sum_{\vv \in \calS} \alpha_{\uu}^* \alpha_{\vv} \braket{\theta_{\uu} \vert \theta_{\vv}}.
    \end{align}
    Thus,
    \begin{align}
        1 &\leq \sum_{\xx \notin \calS} \sum_{\yy \notin \calS} \alpha_{\xx}^* \alpha_{\yy} \braket{\theta_{\xx} \vert \theta_{\yy}}
        \\
        \nonumber
        & = \sum_{\xx \notin \calS} \sum_{\yy \notin \calS} \alpha_{\xx}^* \alpha_{\yy} \braket{\xx \vert M^{\otimes n} \vert \yy}\\
        \nonumber
        &=  \operatorname{Tr}\left[ \left( \sum_{\xx \notin \calS} \sum_{\yy \notin \calS} \alpha_{\xx}^* \alpha_{\yy} \ket{\yy}\bra{\xx}\right) M^{\otimes n} \right]\\
        \nonumber
        &\leq (1+ \cos(\theta))^{n} \cdot \sum_{\xx \notin \calS} \Abs{\alpha_{\xx}}^2,
        \nonumber
    \end{align}
    where \cref{eq:trace_ineq} has been used in the last step. Therefore, 
    \begin{align}
        \sum_{\xx \notin \calS} \Abs{\alpha_{\xx}}^2 \geq  (1+ \cos(\theta))^{-n}.
    \end{align}
    Plugging this in yields the desired result
    \begin{align}
        \Delta(H(\theta)) & = \sin^{2k}(\theta) \left(\frac{1 - \cos(\theta)}{1 + \cos(\theta)}\right)^n.
    \end{align}
\end{proof}

Let us end this section by noting the following two things:
\begin{remark}
    If one substitutes the lower bound on the gap scaling from \cref{theorem:gap_scaling_benjamin} into \cref{thrm:summary_runtimes_gap}, the result suggests an optimal angle of $\theta = \frac{\pi}{2}$. This conclusion, however, stands in contrast with the numerical results reported in Ref.~\cite{benjamin2017}. Moreover, the examples in \cref{sec:Analysis of some restricted input classes} clearly demonstrate that choosing a non-trivial angle can significantly improve the overall runtime of the algorithm. In agreement with the results of Ref.~\cite{benjamin2017}, our findings provide further evidence that the lower bound on the gap derived in \cref{theorem:gap_scaling_benjamin} is not tight. As a note of caution, however, we point out that the numerics in this instance do not capture worst-case performance and can thus only be seen as an indicator for the existence of a better lower bound on the Hamiltonian gap.
\end{remark}

\begin{remark}
If an angle of subspaces $\mu$ is known through the \emph{method of alternating projections} in Section~\ref{sec:alternating_projections}, then this automatically implies a lower bound for the spectral gap $\Delta(\theta,n,m,k)$ of the associated Hamiltonian $H(\theta)$ by virtue of the quantum union bound from \cref{prop:QUB}. In particular,
\begin{align}
    1 - 4\Delta(\theta,n,m,k) \leq \mu\\
    \Rightarrow \Delta(\theta,n,m,k) \geq \frac{1 - \mu}{4}.
\end{align}
\end{remark}

\section{Known worst-case bounds for gaps}
\label{sec:worst-case_bounds_for_gaps}

As highlighted in \cref{thrm:summary_runtimes_gap}, the asymptotic runtime of the algorithm is ultimately determined by how the spectral gap of the rotated Hamiltonian scales. In this section, we provide a brief overview of the existing literature on spectral gaps of Hamiltonians. However, obtaining a tighter bound would directly translate into a better asymptotic runtime. Moreover, for some instances, the gap might scale substantially better than suggested by the worst-case bound (even when it is tight). 

Consider a system described by the family of Hamiltonians $\calH(N)$ where $N$ is the scalable system size. We say that a system is uniformly  \emph{gapped} if, in the large $N$ limit, the gap remains lower bounded by a constant. On the contrary, if the gap vanishes in the large $N$ limit, the system is called \emph{gapless}. A large community of quantum many-body theory is concerned with answering whether a system is gapped or gapless - a question known as the \emph{spectral gap problem}.

It has been shown in Ref.~\cite{cubitt2022} that it is, in general, \emph{undecidable} to solve the spectral gap problem. For this purpose, the authors consider families of nearest neighbor Hamiltonians on a two-dimensional  square lattice of $d$-level quantum systems. However, this does not imply that the spectral gap problem in our setting is undecidable. In fact, many statements are known about the gap scaling of frustration-free Hamiltonians as the variety of statements in the literature showcases.

Most bounds available are for geometrically local Hamiltonians.
Finite-size criteria (a.k.a.\ Knabe-type bounds) 
uplift sufficiently large local gaps to uniform gaps for the infinite system \cite{LemmKnabe,LemmKnabe2}.
Martingale method are related in mindset and 
build an increasing sequence of regions and control the overlaps of their ground-space projectors 
\cite{NachtergaeleBound,ImprovedNachtergaele}. For spin models that can be mapped to non-interacting fermions, also the spectral gap can be efficiently computed.
Let us note that the quantum many-body literature is typically concerned with determining whether a system is gapped or not, rather than bounding how rapidly a gap closes with increasing system size. Moreover, the theorems from above are often restricted to translation-invariant Hamiltonians on a lattice. However, this is not the kind of Hamiltonian we are concerned with in this work.

On the other hand, proven gaps for adiabatic quantum computing (see, e.g.,  Refs.~\cite{farhi2000,vanDam2001,farhi2002,Reichardt2004,farhi2005,Znidaric2006,farhi2010,Altshuler2010,hen2014,werner2023,braida2025}) are often in the form that the gap closes exponentially fast in the system size, which is what we also expect in our setting. However, the Hamiltonians are often very constructed, and it is not so clear how these ideas can be generalized to our setting. Possibly most \emph{useful} in the present context, and presumably at the same time the least known, are perturbations $H=H_0+V$,
where $V$ constitutes a relatively bounded perturbation and $H_0$ is a simple gapped Hamiltonian, ideas that also work for gaps that are closing in the system size.

On a meta-level, we emphasize that understanding the scaling of spectral gaps is crucial for making precise statements about the performance of quantum optimization algorithms. At the same time, standard techniques from quantum many-body theory are often not directly applicable in this setting, as the Hamiltonians arising in optimization problems are typically neither translation-invariant nor geometrically local on a lattice. This discrepancy highlights the need for methods beyond the standard toolbox from quantum many-body theory. One promising direction in this regard is the combination of the \emph{quantum union bound} with the \emph{method of alternating projections}: here, the angle between subspaces provides a mechanism to establish the existence of a spectral gap for a certain family of Hamiltonians.
\end{document}